\documentclass[a4paper,ukenglish,cleveref, autoref, thm-restate,authorcolumns]{lipics-v2019}
\usepackage[utf8]{inputenc}
\usepackage[algo2e,linesnumbered]{algorithm2e}
\usepackage{amsmath}
\usepackage{amsthm}
\usepackage{verbatim}

\nolinenumbers{}
\newcommand{\congest}{\textsc{Congest}}
\newcommand{\local}{\textsc{Local}}
\newcommand{\pl}{\mbox{polylog}}
\newcommand{\pll}{\mbox{polyloglog}}
\newcommand{\eps}{\varepsilon}
\newcommand{\Local}{\textsc{Local}}
\newcommand{\nbhd}{\mathrm{Nbr}}
\newcommand{\dos}{\textsc{DegOrderedSparsify}}

\newif\ifdraft
\newif\ifhideproofs
\draftfalse
\hideproofstrue

\usepackage[textsize=tiny]{todonotes}

\newcommand{\shreyas}[1]{\ifdraft{} \todo[color=red!20]{#1}\fi}

\title{Sample-and-Gather: Fast Ruling Set Algorithms in the Low-Memory MPC Model}

\author{Kishore Kothapalli}{IIIT Hyderabad, India}{kkishore@iiit.ac.in}{}{}
\author{Shreyas Pai}{The University of Iowa, USA}{shreyas-pai@uiowa.edu}{https://orcid.org/0000-0003-2409-7807}{}
\author{Sriram V.~Pemmaraju}{The University of Iowa, USA}{sriram-pemmaraju@uiowa.edu}{}{}

\authorrunning{K. Kothapalli, S. Pai, and S. V. Pemmaraju}
\Copyright{Kishore Kothapalli, Shreyas Pai, and Sriram V.~Pemmaraju}

\date{}
\EventEditors{Nitin Saxena and Sunil Simon}
\EventNoEds{2}
\EventLongTitle{40th IARCS Annual Conference on Foundations of Software Technology and Theoretical Computer Science (FSTTCS 2020)}
\EventShortTitle{FSTTCS 2020}
\EventAcronym{FSTTCS}
\EventYear{2020}
\EventDate{December 14--18, 2020}
\EventLocation{BITS Pilani, K K Birla Goa Campus, Goa, India (Virtual Conference)}
\EventLogo{}
\SeriesVolume{182}
\ArticleNo{12}

\begin{document}

\maketitle
\begin{abstract}
  Motivated by recent progress on symmetry breaking problems such as maximal independent set (MIS)
and maximal matching in the low-memory Massively Parallel Computation (MPC) model
(e.g., Behnezhad et al.~PODC 2019; Ghaffari-Uitto SODA 2019), we investigate the complexity of ruling set problems in this model.
The MPC model has become very popular as a model for large-scale distributed computing and it comes with the constraint that the memory-per-machine is strongly sublinear in the input size. For graph
problems, extremely fast MPC algorithms have been designed assuming $\tilde{\Omega}(n)$ memory-per-machine,
where $n$ is the number of nodes in the graph (e.g., the $O(\log\log n)$ MIS algorithm of Ghaffari et al., PODC 2018). However, it has proven
much more difficult to design fast MPC algorithms for graph problems in the
\textit{low-memory} MPC model, where the memory-per-machine is  restricted to being strongly sublinear in the number of nodes, i.e., $O(n^\eps)$ for $0 < \eps
< 1$.

In this paper, we present an algorithm for the 2-ruling set problem, running in $\tilde{O}(\log^{1/6} \Delta)$
rounds whp, in the low-memory MPC model.
We then extend this result to $\beta$-ruling sets for any integer $\beta > 1$. Specifically, we show that a $\beta$-ruling set can be computed in the low-memory
MPC model with $O(n^\eps)$ memory-per-machine in $\tilde{O}(\beta \cdot \log^{1/(2^{\beta+1}-2)} \Delta)$
rounds, whp. From this it immediately follows that a $\beta$-ruling set for $\beta = \Omega(\log\log\log \Delta)$-ruling set can be computed in
in just $O(\beta \log\log n)$ rounds whp. The above results assume a total memory of $\tilde{O}(m + n^{1+\eps})$. We also present algorithms for $\beta$-ruling sets in the low-memory MPC model assuming that the total memory over all machines is restricted to $\tilde{O}(m)$.
These algorithms are all substantially faster than the Ghaffari-Uitto $\tilde{O}(\sqrt{\log \Delta})$-round MIS algorithm in the low-memory MPC model.

All our results follow from a \textit{Sample-and-Gather Simulation Theorem} that shows how random-sampling-based \textsc{Congest} algorithms can be efficiently simulated in the low-memory MPC model.
We expect this simulation theorem to be of independent interest beyond the
ruling set algorithms derived here.
 \end{abstract}

\section{Introduction}
There has been considerable recent progress in the design and study of large-scale distributed computing models that are closer to reality, yet mathematically tractable.
Of these, the \textit{Massively Parallel Computing (MPC)} model~\cite{Karloff2010,YaroslavtsevVadapalliICML2018} has gained significant attention due to its flexibility and its ability to closely model existing distributed computing frameworks used in practice such as MapReduce~\cite{DeanG08},
Spark~\cite{spark}, Pregel~\cite{pregel}, and Giraph~\cite{ChingetalVLDB2015}.

The MPC model is defined by a set of machines, each having at most $S$ words of memory. The machines are connected to each other via an all-to-all communication network. 
Communication and computation in this model are synchronous.
In each round, each machine receives up to $S$ words from other machines, performs local computation, and sends up to $S$ words to other machines. 
The key characteristic of the MPC model is that both the memory upper bound $S$ and the number of machines used are assumed to be strongly sublinear in the input size $N$, i.e., bounded by $O(N^{1-\eps})$ for some constant $\eps$, $0 < \eps < 1$.
This characteristic models the fact that in modern large-scale computational problems the input is too large to fit in a single machine and is much larger than the number of available machines. 

Even though the MPC model is relatively new, a wide variety of classical graph problems have been studied in this model.
This stream of research includes the design of fast algorithms~\cite{Assadi2018a,BehnezhadetalPODC2019,Czumaj2017,Czumaj2019,Ghaffari2018} as well as lower bound constructions~\cite{Charikar2020,Ghaffari2019,RoughgardenVWJACM2018}.
A particular, though not exclusive, focus of this research has been on \textit{symmetry breaking} problems such as maximal independent set (MIS)~\cite{BehnezhadetalPODC2019,Ghaffari2018,Ghaffari2020}, maximal matching~\cite{Behnezhad2019}, and $(\Delta+1)$-coloring~\cite{ChangetalPODC2019,Assadi2018b}, along with related graph optimization problems such as minimum vertex cover and maximum matching.

For graph problems, the input size is $\tilde{O}(m+n)$ where $m$ is the number of edges and $n$ is the number of nodes of the input graph. Thus, $O((m+n)^{1-\eps})$, for some constant $\eps$, $0 < \eps < 1$, is an upper bound on both the number of machines that can be used and the size $S$ of memory per machine.
It turns out that the difficulty of graph problems varies significantly based on how $S$ relates to
the number of nodes (\(n\)) of the input graph. Specifically, three regimes for $S$ have been considered in the literature.
\begin{itemize}
\item \textbf{Strongly superlinear memory (\(S = O(n^{1 + \eps})\)):} 
For this regime to make sense in the MPC model, the input graph needs to be
highly dense, i.e., \(m \gg S \gg n\) such that \(S\) is strongly sublinear in \(m\). 
Even though the input graph is dense, the fact that each machine has $O(n^{1+\eps})$ local memory makes this model quite powerful. For example, in this model, problems such as minimum spanning tree, MIS, and \(2\)-approximate minimum vertex cover, all have $O(1)$-round algorithms~\cite{Karloff2010,HarveyetalSPAA2018}.

\item \textbf{Near-linear memory
(\(S=\tilde{O}(n)\)):} Problems become harder in this regime, but symmetry breaking problems such as MIS, vertex cover, and maximal matching can still be solved in \(O(\log \log n)\) rounds~\cite{Czumaj2017,Assadi2017a,Ghaffari2018a,GhaffariJNSPAA20}.
Furthermore, recently Assadi, Chen, and Khanna~\cite{Assadi2018b} presented an $O(1)$-round algorithm for \((\Delta+1)\)-vertex coloring.

\item \textbf{Strongly sublinear memory (\(S = O(n^{\eps})\)):} 
Problems seem to get much harder in this regime and whether there are sublogarithmic-round algorithms for certain graph problems in this regime is an important research direction.
For example, it is conjectured
that the problem of distinguishing if the input graph is a single
cycle vs two disjoint cycles of length \(n/2\) requires \(\Omega(\log n)\) rounds \cite{YaroslavtsevVadapalliICML2018,Ghaffari2019}.
However, even in this regime, Ghaffari and Uitto~\cite{Ghaffari2018} have recently shown that MIS does have a sublogarithmic-round algorithm, running in  $\tilde{O}(\sqrt{\log \Delta})$ rounds, where $\Delta$ is the maximum degree of the input graph.
This particular result serves as a launching point for the results in this paper. 
\end{itemize}

The MIS problem has been called 
``a central problem in the area of locality in distributed computing''
(2016 Dijkstra award citation).
Starting with the elegant, randomized MIS algorithms from the mid-1980s by Luby~\cite{LubySICOMP1986} and by Alon et al.~\cite{AlonetalJAlg1986}, several decades of research has now been devoted to designing MIS algorithms in various models of parallel and distributed computing (e.g., PRAM, \textsc{Local}, \textsc{Congest}, \textsc{Congested-Clique}, and MPC).
A \textit{ruling set} is a natural relaxation of an MIS and considerable research has been devoted to solving the 
ruling set problem in different models of distributed 
computation as well~\cite{BEPS16,KothapalliP12,BishtKP13,GhaffariSODA2016}.
An \textit{$(\alpha, \beta)$-ruling set} of a graph $G = (V, E)$ is a  subset $S \subseteq V$ such that (i) every pair of nodes in $S$ are at distance at least $\alpha$ from each other and (ii) every node in $V$ is at distance at most $\beta$ from some node in $S$.
An MIS is just a \((2, 1)\)-ruling set.
Research on the ruling set problem has focused on the question of how much faster distributed ruling set algorithms can be relative to MIS algorithms and whether there is a provable separation in the distributed complexity of these problems in different 
models of distributed computing.
For example, in the \Local\ model\footnote{The \Local{} model is a synchronous, message passing model of distributed computation \cite{Linial92,peleg00} with unbounded messages. See Section \ref{subsection:notation} for definitions of related models of computation.}, Kuhn, Moscibroda, and Wattenhofer \cite{KMWPODC2004,KMWJACM2016} show an $\Omega\left(\min\left\{\frac{\log \Delta}{\log\log \Delta}, \frac{\log \Delta}{\log\log \Delta}\right\}\right)$ lower bound for MIS, even for randomized algorithms. However, combining the recursive sparsification procedure of Bisht et al.~\cite{BishtKP13} with the improved MIS algorithm of Ghaffari \cite{GhaffariSODA2016} and the recent deterministic network decomposition algorithm of Rozhon and Ghaffari \cite{RozhonGSTOC20}, it is possible to compute $\beta$-ruling sets in $O(\beta \log^{1/\beta} \Delta + \pll(n))$ rounds, thus establishing a separation between these problems, even for $\beta = 2$, in the \Local\ model.
In this paper, we are interested only in $(2, \beta)$-ruling sets and so as a short hand, we drop the first parameter ``2'' and call these objects $\beta$-ruling sets.
As a short hand, we will use \textit{low-memory MPC model} to refer to the \textit{strongly sublinear memory MPC model}.
As mentioned earlier, Ghaffari and Uitto~\cite{Ghaffari2018} recently presented an algorithm that solves MIS in the low-memory MPC model in $\tilde{O}(\sqrt{\log \Delta})$ rounds.  However, nothing more is known about the \(2\)-ruling set problem in this model and the fastest \(2\)-ruling set algorithm in the low-memory MPC model is just the above-mentioned MIS algorithm.
This is in contrast to the situation in the linear-memory MPC model.
In this model, the fastest algorithm for solving MIS runs in $O(\log\log n)$ rounds~\cite{Ghaffari2018a}, whereas the fastest \(2\)-ruling set algorithm runs in $O(\log\log\log n)$ rounds~\cite{HegemanPSarxiv2014}.
This distinction between the status of MIS and \(2\)-ruling sets in the linear-memory MPC model prompts the following related questions. 
\begin{quote}
\textit{
Is it possible to design an $o(\sqrt{\log \Delta})$-round, \(2\)-ruling set algorithm in the low-memory MPC model? Could we in fact design 2-ruling set algorithms in the low-memory MPC model that run in $O(\pll(n))$ rounds?} 
\end{quote}

\subsection{Main Results}\label{subsection:mainResults}
We make progress on the above question via the following results proved in this paper. 
\begin{enumerate}
\item
We show (in Theorem~\ref{theorem:2rslow} part (i)) that a \(2\)-ruling set of a graph $G$ can be computed in $\tilde{O}(\log^{1/6} \Delta)$ rounds in the low-memory MPC model.
We generalize this result to $\beta$-ruling sets, for $\beta \ge 2$ (in Theorem~\ref{thm:betarsMPC} part (i)), and show that a $\beta$-ruling set of a graph $G$ can be computed in $\tilde{O}(\log^{1/(2^{\beta+1}-2)} \Delta)$ rounds in the low-memory MPC model.
These algorithms are substantially faster than the MIS algorithm~\cite{Ghaffari2018} for the low-memory MPC model.
The inverse exponential dependency on $\beta$ in the running time of the $\beta$-ruling set algorithm is worth noting. This dependency implies that for any 
$\beta = \Omega(\log\log\log \Delta)$, we can compute a $\beta$-ruling set in only $O(\beta \cdot \pll(n))$ rounds.
This is in contrast to the situation in the \Local\ model; using the $O(\beta \cdot \log^{1/\beta} \Delta + \pll(n))$-round $\beta$-ruling set algorithm in the \Local\ model mentioned earlier, one can obtain an $O(\pll(n))$-round algorithm only for $\beta = \Omega(\log\log \Delta)$.

\item Even though the above-mentioned results are in the low-memory MPC model, they assume no restrictions on the total memory used by all the machines put together.
Specifically, we obtain the above results allowing a total of $\tilde{O}(m + n^{1+\eps})$ memory.
Note that the input uses $\tilde{O}(m)$ memory and thus these algorithms make use of $\tilde{O}(n^{1+\eps})$ extra total memory.
If we place the restriction that the total memory cannot exceed the input size, i.e., $\tilde{O}(m)$, then we get slightly weaker results.
Specifically, we show (in Theorem~\ref{theorem:2rslow} part (ii)) that a \(2\)-ruling set can be computed in $\tilde{O}(\log^{1/4} \Delta)$ rounds in the low-memory MPC model using
$\tilde{O}(m)$ total memory.
Additionally, we show (in Theorem~\ref{thm:betarsMPC} part (ii)) that a $\beta$-ruling set, for any $\beta \ge 2$, can be computed in $\tilde{O}(\log^{1/2\beta} \Delta)$ rounds in the low-memory MPC model using $\tilde{O}(m)$ total memory.
Note that even though these results are weaker than those we obtain in the setting where total memory is unrestricted, these algorithms are much faster than the $\tilde{O}(\sqrt{\log \Delta})$-round,
low-memory MPC model algorithm for MIS that uses $\tilde{O}(m)$ total memory \cite{Ghaffari2018}.
\end{enumerate}

\noindent
\textbf{Technical Contributions.} We obtain all of these results by applying new Simulation Theorems (Theorems~\ref{theorem:sampleAndGather} and~\ref{theorem:mTotalMemory}) that we develop and prove.
These Simulation Theorems provide a general method for deriving fast MPC algorithms
from known distributed algorithms in the \congest\ model\footnote{The \congest\ model \cite{peleg00} is similar to the \local\ model except that in the \congest\ model there is an $O(\log n)$ bound on the size of each message.} 
and they form the main technical contribution of this paper.

A well-known technique \cite{Ghaffari2017,Ghaffari2018,HegemanPSarxiv2014,ParterYogevDISC2018} for designing fast algorithms in ``all-to-all'' communication models such as MPC is the following ``ball doubling'' technique. 
Informally speaking, if for every node $v$ we know the state of the $k$-neighborhood around node $v$, then by exchanging this information, ideally in $O(1)$ rounds, it is possible to learn the state of the $2k$-neighborhood around each node.
Thus, having learned the state of an $\ell$-neighborhood around each node $v$ in $O(\log \ell)$ rounds, it is possible to simply use local computation at each node to ``fast forward'' the algorithm by $\ell$ rounds, without any further communication.
In this manner, a phase consisting of $\ell$ rounds in the \congest\ model can be compressed into $O(\log \ell)$ rounds in the MPC model.
This description of the ``ball doubling'' technique completely ignores the main obstacle to using this technique: the $k$-neighborhoods around nodes may be so large that bandwidth constraints of the communication network may disallow rapid exchange of these $k$-neighborhoods.

Our main contribution is to note that in many randomized, distributed algorithms in the \congest\ model, there is a natural \textit{sparsification} that occurs, i.e., in each round a randomly sampled subset of the nodes are active, and the rest are silent.
This implies that the $k$-neighborhoods that are exchanged only need to involve sparse subgraphs induced by the sampled nodes.
A technical challenge we need to overcome is that the subgraph induced by sampled nodes is not just from the next round, but from the $\ell$ future rounds; so we need to be able to estimate which nodes will be sampled in the future.
On the basis of this idea, we introduce the notion of $\alpha$-sparsity of a randomized \congest\ algorithm, for a parameter $\alpha$; basically smaller the $\alpha$ greater the sparsification induced by random sampling.
We present \textit{Sample-and-Gather} Simulation Theorems in which, roughly speaking, an $R$-round \congest\ algorithm is simulated in $\tilde{O}(R/\sqrt{\log_\alpha n})$ rounds (respectively, $\tilde{O}(R/\sqrt{\log_\alpha \Delta})$ rounds) in the low-memory MPC model, where the total memory is $\tilde{O}(m + n^{1+\eps})$ (respectively, $\tilde{O}(m)$).

Our Simulations Theorems are inspired by a Simulation Theorem due to Behnehzhad et al.~\cite[Lemma 5.5]{BehnezhadetalPODC2019}. Using their Simulation Theorem, an $R$-round state-congested algorithm can be simulated in (roughly) $R/\log_\Delta n$ low-memory MPC rounds. In contrast, our Simulation Theorem (Theorem \ref{theorem:sampleAndGather}) yields a running time of (roughly) $R/\sqrt{\log_\alpha n}$, where $\alpha$ is a sparsity parameter. When the input graph has high degree, but the state-congested algorithm samples a very sparse subgraph (i.e., $\alpha$ is small) then our Simulation Theorems provide a huge advantage over the Behnehzhad et al.~Simulation Theorems.

To obtain our results for ruling sets, we apply the Sample-and-Gather Simulation Theorems to the sparsification procedure of Kothapalli and Pemmaraju~\cite{KothapalliP12} and Bisht et al.~\cite{BishtKP13} and to the sparsified MIS algorithm of Ghaffari~\cite{Ghaffari2017}. 
We note that by applying the Sample-and-Gather Simulation Theorems to the sparsified MIS algorithm of Ghaffari~\cite{Ghaffari2017}, we recover the Ghaffari-Uitto low-memory MPC algorithm for MIS \cite{Ghaffari2018}, built from scratch.
We believe that the Sample-and-Gather Theorems will be of independent interest because they simplify the design of fast MPC algorithms.

\ifhideproofs{}
\subsection{Other Related Work}\label{subsection:relatedWork}
The MPC model has received a lot of interest on problems other than the ones mentioned in the previous section \cite{Goodrich2011,LattanziMSVSPAA11,BeameKS17,KoutrisBS16,BeameKS14,AhnGSPAA15,ImBSSTOC17,BoroujeniEGHS18}. For example~\cite{AndoniNOY14,AssadiKZPODC19,Inamdar2018} consider clustering problems and~\cite{AndoniSZ20,AssadiSWPODC19,Andoni2018,ShankhaBiswas2020} look at distance computation problems like minimum spanning tree, shortest paths, and spanners.

There has been some progress in recent years in simulating distributed algorithms from one model of computation to another. Karloff et al.~\cite{Karloff2010} show how to simulate certain PRAM algorithms in the MPC (or Map-Reduce) model. Hegeman and Pemmaraju \cite{HegemanPTCS15} show that algorithms designed in the \textsc{Congested-clique} model can be simulated in the Map-Reduce model~\cite{Karloff2010}. The upper bounds shown by Klauck et al. \cite{KlauckNPRSODA15} also are the result of converting algorithms designed in the \textsc{Congest} model to algorithms in the $k$-machine model. The $k$-machine model \cite{KlauckNPRSODA15} is a recent distributed computing model that consists of a set of $k$ pairwise interconnected machines with link bandwidth of $B$ per round.  Konrad et al. \cite{Konrad} show that algorithms in the beeping model, a model that is inspired in part by communication in biological processes, can be simulated to run in the $k$-machine model. Behnezhad et al. \cite{BehnezhadDH18} show that algorithms designed in the \textsc{Congested-Clique} model can be simulated in the semi-MPC model. \shreyas{Need conference version of \cite{BehnezhadDH18}}
\fi

\subsection{Technical Preliminaries}\label{subsection:notation}
\textbf{Models.} In the \textsc{Congest} model \cite{peleg00} a communication network is abstracted as an $n$-node graph. In synchronous rounds each node can send a $O(\log n)$ bit message to each of its neighbors. The \emph{complexity} is the number of rounds until each node has computed its output, e.g., whether it belongs to an MIS or not.
The \textsc{Congested-Clique} model is similar to the \textsc{Congest} model, but nodes  can send $O(\log{n})$-bits messages to all other nodes, not only to its neighbors in the input graph $G$~\cite{Peleg03}. The \textsc{Local} model~\cite{Linial92} is the same as the \textsc{Congest} model, except the message sizes can be unbounded.

\noindent
\textbf{MPC simulations.} 
In the low-memory MPC model, even a single round of a \congest\ algorithm in which every node sends a message to every neighbor, is hard to simulate.
This is because the degree of a node could be larger than the memory volume $n^\epsilon$ of a machine.
To deal with this issue, we first assume that a node $v$ with $\deg(v) > n^\epsilon$ is split into copies that are distributed among different machines and we have a virtual
$O(1/\eps)$-depth balanced tree on these copies of $v$.
The root of this tree coordinates communication between $v$ and its neighbors in the input graph. By itself, this is insufficient because information from $v$'s neighbors cannot travel up $v$'s tree without running into a memory bottleneck.
However, if computation at each node can be described by a \textit{separable} function, then this is possible. 
The following definition of separable functions captures functions such as $\max$, $\min$, sum, etc. This issue and the proposed solution have been discussed in \cite{Ghaffari2018,BehnezhadetalPODC2019}.
\begin{definition}\label{def:separable}
Let $f: 2^{\mathbb{R}} \rightarrow \mathbb{R}$ denote a set function. We call $f$ \emph{separable} iff for any set of reals $A$ and for any $B \subseteq A$, we have $f(A) = f\big(f(B), f(A \setminus B)\big)$.
\end{definition}
The following lemma \cite{BehnezhadetalPODC2019} shows that it is possible to compute the value of a separable function $f$ on each of the nodes in merely $O(1/\eps)$ rounds. 
The bigger implication of this lemma is that a single round of a \congest\ algorithm can be simulated in $O(1/\eps)$ low-memory MPC rounds. 

\begin{lemma} \label{lem:generalmaxload}
  Suppose that on each node $v \in V$, we have a number $x_v$ of size $O(\log n)$ bits and let $f$ be a separable function. There exists an algorithm that in $O(1/\eps)$ rounds of MPC, for every node $v$, computes $f(\{x_u \, | \, u\in \nbhd(v) \})$ whp in the low-memory MPC model with $\tilde{O}(m)$ total memory.
\end{lemma}

\noindent
\textbf{Graph-theoretic notation.} For a node $v\in V$ we denote its non-inclusive neighborhood in $G$ by \(\nbhd(v)\). Moreover, we define \(\nbhd^{+}(v) = \nbhd(v) \cup \{v\}\), \(\nbhd(S) = \bigcup_{v \in S}\nbhd(v)\), \(\nbhd^{+}(S) = \bigcup_{v \in S}\nbhd^{+}(v)\).

\ifhideproofs{}
\fi{} 
\section{The Sample-and-Gather Simulation}\label{section:SampleAndGather}
Our simulation theorems apply to a subclass of \congest\ model algorithms
called \textit{state-congested} algorithms \cite{BehnezhadetalPODC2019}.
\begin{definition}
An algorithm in the \congest\ model is said to be \textit{state-congested} if
\begin{itemize}
    \item[(i)] by the end of round $r$, for any $r$, at each node 
    $v$, the algorithm stores a state $\sigma_r(v)$ of size $O(\deg(v) \pl(n))$ bits, i.e., an average of $O(\pl(n))$ bits per neighbor. The initial state $\sigma_0(v)$ of each node $v$ is its \texttt{ID}.
    Furthermore, we can update the state at each node $v$ in each round $r$ using an additional temporary space of size $O(\deg(v) \cdot \pl (n))$ bits.
    \item[(ii)] The states of the nodes after the last round of the algorithm
    are sufficient in determining, collectively, the output of the
    algorithm.
\end{itemize}
\end{definition}

\noindent
A key feature of a state-congested algorithm is that the local state at each node stays bounded in size
throughout the execution of the algorithm.

We inductively design a fast low-memory MPC algorithm that simulates a given state-congested algorithm. For this purpose, we start by assuming that we have a state-congested, possibly randomized, algorithm $Alg$, whose first $t$ rounds have
been correctly simulated in the low-memory MPC model. Our goal now is to simulate a \textit{phase} consisting of the next $\ell$ rounds of $Alg$, i.e., rounds $t+1, t+2, \ldots, t+\ell$, in just $O(\log \ell)$ low-memory MPC rounds.
We categorize each node $u$ in a round $\tau$, $t+1 \le \tau \le t+\ell$,
based on its activity in round $\tau$. Specifically, a node $u$ is  a
\textit{sending node} 
in round $\tau$ if sends at least one message in round $\tau$. Moreover, a node is called a \textit{sending-only node} if it does not update its state in round $\tau$.

Consider a node $u$ at the start of the phase we want to compress. Since this is immediately after round $t$, node
$u$ knows its local state $\sigma_t(u)$.
Let $p_{t+1}(u)$ denote the probability that node $u$ is a sending node in round $t+1$. We call this the
\textit{activation probability} of node $u$ in round $t+1$. 
Also, for any node $v$, let $A_{t+1}(v) := \sum_{u \in Nbr(v)} p_{t+1}(u)$ denote the \textit{activity level} in $v$'s neighborhood in round $t+1$.
Note that $p_{t+1}(u)$ is completely determined by $\sigma_t(u)$ and so
node $u$ can locally calculate $p_{t+1}(u)$ after round $t$.
In order to simulate rounds $t+1, t+2, \ldots, t+\ell$ in a compressed
fashion in the MPC model, every node $u$ needs to know the probability of it being a sending node in each of these rounds. But, rounds $t+2, t+3, \ldots, t+\ell$ are in the future and so node $u$, using current knowledge, can only \textit{estimate} an 
upper bound $\tilde{p}_\tau(u)$ on the probability that it will be a sending node
in round $\tau$, $t+2 \le \tau \le t+\ell$.

\ifhideproofs{}
To do this estimation, node $u$ considers all \textit{feasible} current global states $\Pi$. 
As a short hand, we will use \textit{round-$\tau$ local state} 
(respectively, \textit{round-$\tau$ global state}) to denote a local (respectively, global) state immediately after round $\tau$.
Now note that from $u$'s point of view, for a global state $\Pi$ to be a feasible round-$t$ global state, the local state of $u$ in $\Pi$ should equal $\sigma_t(u)$.
Further note that if $u$ knows an upper bound on the number of nodes in the network, this set
of global states is finite.
For each such global state $\Pi$, let $S_{\tau-1}(u, \Pi)$ denote the collection 
of all round-$(\tau-1)$ local states of node $u$ reachable 
from the round-$t$ global state $\Pi$.
Given that $Alg$ is randomized, its execution induces a probability distribution over $S_{\tau-1}(u, \Pi)$. Let $S^{hp}_{\tau-1}(u, \Pi)$ denote an arbitrary high probability subset of $S_{\tau-1}(u)$.
Then, we define $\tilde{p}_\tau(u, \Pi)$
as the maximum probability of node $u$ being a send-only node in any state in $S^{hp}_{\tau-1}(u, \Pi)$.
Finally, we define $\tilde{p}_\tau(u) := \max_{\Pi} p_{\tau}(u, \Pi)$ as the worst case estimate, over all feasible global states.
This definition of $\tilde{p}_\tau(u)$ implies that whp\footnote{We use ``whp'' as short for ``with high probability'' which
refers to the probability that is at least $1 - 1/n^c$ for $c \ge 1$.} in any execution starting in
a round-$t$ global state in which the local state of node $u$ is $\sigma_t(u)$,
the probability that node $u$ will be a sending node in round $\tau$ is
bounded above by $\tilde{p}_\tau(u)$.
This definition of $\tilde{p}_\tau(u)$ holds for all rounds $\tau = t+2, t+3, \ldots, t+\ell$. For round $\tau = t+1$, we simply set $\tilde{p}_{t+1}(u) := p_{t+1}(u)$, i.e., the estimated activation probability in round $t+1$ is the actual activation probability.
One final remark about these probability estimates $\tilde{p}_\tau(u)$ is that they 
can all be computed by node $u$, using just local knowledge.
In theory, this may take super-polynomial time, which is allowed in the \congest\ model.
But in practice, as can be seen from the applications in Section \ref{sec:2rs}, estimating these probabilities is a polynomial-time computation.
\fi{}

Let $\tilde{p}_{t+1}(u) = p_{t+1}(u)$ for any node $u$.
For any $\tau$, $t+1 \le \tau \le t+\ell$, for any node $v$, let $\tilde{A}_\tau(v) := \sum_{u \in Nbr(v)} \tilde{p}_\tau(u)$ denote the \textit{estimated activity level} in node $v$'s neighborhood in round $\tau$.
Note that for the first round round in the phase, $\tau = t+1$, the estimated and actual activity levels are identical. 
Finally, let $\tilde{A}_\tau$ be the maximum $\tilde{A}_\tau(v)$, where the maximum is taken
over all nodes $v$ that are not sending-only nodes.

\begin{lemma} 
\label{lemma:sampleAndGather}
Suppose $\ell$ is such that
\begin{equation}
\label{equation:ellBound}
\left(\sum_{\tau = t+1}^{t+\ell} \tilde{A}_\tau \log n\right)^{\ell} \le O(n^{\eps/2}).
\end{equation}
Then the next phase of the algorithm $Alg$ consisting of rounds $t+1, t+2, \ldots, t+\ell$ can be simulated in $O(\log \ell)$ 
rounds in the low-memory MPC model with $\tilde{O}(m + n^{1+\eps})$ total
memory.
\end{lemma}
\ifhideproofs{}\begin{proof}
Simulating rounds $t+1, t+2, \ldots, t+\ell$ of algorithm $Alg$ is equivalent to computing the state $\sigma_{t+\ell}(v)$ for every node $v \in V$.
We use the 2-step algorithm below to do this computation.
First, we introduce some notation. Let $B_G(v, \ell)$ denote the labeled subgraph of
$G$, induced by nodes that are at most $\ell$ hops from $v$ in $G$ and in which each node $u$ is labeled with its local state $\sigma_t(u)$ after round $t$.
\begin{description}
\item[Step 1:] For each node $v \in V$, designate a distinct machine $M_v$ at which we gather a ``sampled'' subgraph $S_G(v, \ell)$ of $B_G(v, \ell)$. The definition of 
$S_G(v, \ell)$ is provided below.
\item[Step 2:] Using the subgraph $S_G(v, \ell)$, machine $M_v$ locally simulates rounds $t+1, t+2, \ldots, t+\ell$ of $Alg$ and computes $\sigma_{t+\ell}(v)$. 
\end{description}

In the rest of the proof, we will first define the subgraph $S_G(v, \ell)$. We will then show in Claim \ref{claim:SGsuffices} that using this subgraph, it is possible for machine $M_v$ to locally simulate rounds $t+1, t+2, \ldots, t+\ell$ of $Alg$.
We then show in Claim \ref{claim:sizeBound} that assuming $\ell$ satisfies (\ref{equation:ellBound}), the size of $S_G(v, \ell)$ is
$O(n^\eps)$ whp.
Finally, in Claim \ref{claim:gatheringTime}, we show that the subgraph $S_G(v, \ell)$ can be gathered at each machine $M_v$ in parallel in $O(\log \ell)$ rounds.
These claims together complete the proof of the lemma.

Each node $u \in V$ generates a sequence of uniformly distributed random bits $r^1_\tau(u)$, $r^2_\tau(u)$, $\ldots$, $r^{c \cdot \log n}_\tau(u)$ for a large enough constant $c$.
These bits are designated for round 
$\tau$, $t+1 \le \tau \le t+\ell$ and they serve two purposes: (i) they are used to randomly sample $u$
based on the estimate $\tilde{p}_\tau(u)$ that $u$ will be 
a sending node in round $\tau$, and (ii) they are used to simulate $u$'s actions in round $\tau$.
It is important that the same bits be used for both purposes so that there is consistency in $u$'s random actions.
Specifically, $u$ constructs a real number $R_\tau(u)$ that is uniformly distributed over $\{i/2^{c\log n} \mid 0 \le i < c\log n\}$ using these bits.
Node $u$ adds these $O(\ell \cdot \log n)$ bits to its local state after round $t$, $\sigma_t(u)$.
Node $u$ then \textit{marks itself for round $\tau$} 
if $R_\tau(u) \le p_\tau(u)$. If a node $u$ marks itself
for  a round $\tau$ it means that in $u$'s estimate
after round $t$, $u$ will be a sending node in round $\tau$.
Further, node $u$ is \textit{marked} if it is marked for round $\tau$ for 
any $\tau$, $t+1 \le \tau \le t+\ell$.
The ``sampled'' subgraph $S_G(v, \ell)$ is the subgraph of $B_G(v, \ell)$ induced by $v$ along with all nodes $u$ in $B_G(v, \ell)$ that are marked.

\begin{claim}
\label{claim:SGsuffices}
For any node $v \in V$, information in  $S_G(v, \ell)$ is enough
to locally compute $\sigma_{t+\ell}(v)$.
\end{claim}
\ifhideproofs{}\begin{proof}
We prove this claim inductively. Specifically, we prove the following:
\begin{quote}
For any $i$, $0 < i \le \ell$,
in addition to knowing $S_G(v, \ell)$,
if we know the states $\sigma_{t+\ell-i}(u)$ for all $u \in S_G(v, i)$ then we can compute the
states $\sigma_{t+\ell-i+1}(u)$ for all $u \in S_G(v, i-1)$.
\end{quote}
The premise of this statement is true for $i = \ell$ because $S_G(v, \ell)$ contains the round-$t$ local states $\sigma_t(u)$ for all $u \in S_G(v, \ell)$.
For $i=1$ this claim is equivalent to saying that in addition to $S_G(v, \ell)$, if we know $\sigma_{t+\ell-1}(u)$ for all neighbors of $v$ in $S_G(v, \ell)$ then we can compute $\sigma_{t+\ell}(v)$. This is what we need to show.

To be able to compute $\sigma_{t+\ell-i+1}(u)$ for any $u$ in $S_G(v, i-1)$, we need to know 
the round-$(t+\ell-i)$ local states $\sigma_{t+\ell-i}(w)$ for all neighbors $w$ of $u$ that are sending nodes in round $t+\ell-i$. 
With high probability, the probability $p_w$ that
a neighbor $w$ of $u$ sends messages in round $t+\ell-i$ is upper bounded by the estimate $\tilde{p}_{t+\ell-i}(w)$ that $w$ computed after round $t$.
Node $w$ sends messages in round $t+\ell-i$ if $R_{t+\ell-i}(w) \le p_w$.
Since $p_w \le \tilde{p}_{t+\ell-i}(w)$, we know that $R_{t+\ell-i}(w) \le \tilde{p}_{t+\ell-i}(w)$ and therefore $w$ is marked and included in $S_G(v, \ell)$.
Also note that since $u \in S_G(v, i-1)$ and $w$ is a neighbor of $u$, we see that $w \in S_G(v, i)$.
Thus any node $w$ that sends a message to node $u$ in round $t+\ell-i$ 
belongs to $S_G(v, i)$ and by the hypothesis of the inductive claim we know $\sigma_{t+\ell-i}(w)$.
With the knowledge of $\sigma_{t+\ell-i}(w)$, we can simulate round $t+\ell-i+1$
at each node $w$, using the random real $R_{t+\ell-i+1}(w)$ to execute any random
actions $w$ may take. 
Then using the message received by $u$ from all such neighbors $w$ in round $t+\ell-i+1$, we can update $u$'s local state, thus computing $\sigma_{t+\ell-i+1}(u)$.
\end{proof}\fi{}

\begin{claim}
\label{claim:sizeBound}
For any node $v \in V$, the size of $S_G(v, \ell)$ is at most $\left(\sum_{\tau = t+1}^{t+\ell} \tilde{A}_\tau \log n\right)^{\ell}$ whp.
\end{claim}
\ifhideproofs{}\begin{proof}
Consider an arbitrary $v \in V$ and $u \in B_G(v,\ell)$ and a round $t+1 \le \tau \le t+\ell$. Node $u$ is marked for round $\tau$ with probability $p_\tau(u)$. 
Recalling that $Nbr(u)$ denotes the set of neighbors of $u$ in $G$, we see that expected number of neighbors of
$u$ marked for round $t+1 \le \tau \le t+\ell$ is at most
$$\sum_{w \in Nbr(u)} p_\tau(w) \le \tilde{A}_\tau(u) \le \tilde{A}_\tau.$$
Furthermore, since neighbors of $u$ are marked for round $\tau$ independently, by Chernoff bounds we see that the number of neighbors that $u$ has in $S_G(v, \ell)$ that are marked for round $\tau$ is $\tilde{A}_\tau \log n$ whp.
By the union bound this means that the number of neighbors that $u$ has in $S_G(v, \ell)$ is $\sum_{\tau=t+1}^{t+\ell} \tilde{A}_\tau \log n$ whp.
From this it follows that the size of $S_G(v, \ell)$ is $\left(\sum_{\tau = t+1}^{t+\ell} \tilde{A}_\tau \log n\right)^{\ell}$.
\end{proof}\fi{}

\begin{claim}
\label{claim:gatheringTime}
For every node $v \in V$, the graph $S_G(v, \ell)$ can be gathered at $M_v$ in at most $O(\log \ell)$ rounds.
\end{claim}
\ifhideproofs{}\begin{proof}
Here we use the ``ball doubling'' technique that appears in a number of papers on algorithms in ``all-to-all'' communication models (e.g., \cite{Ghaffari2017,Ghaffari2018,HegemanPSarxiv2014,ParterYogevDISC2018}). 
Suppose that each machine $M_v$  knows $S_G(v, i)$ for some $0 \le i \le \ell/2$.
Each machine $M_v$ then sends $S_G(v, i)$ to machine $M_u$ for every node $u$ in $S_G(v, i)$.
After this communication is completed, each machine $M_v$ can construct $S_G(v, 2i)$ from the information it has received because $S_G(v, 2i)$ is contained in the union of $S_G(u, i)$ for all $u$ in $S_G(v, i)$. 

We now argue that this communication can be performed in $O(1)$ rounds. First, note that the size of $S_G(v, i)$ is bounded above by $O(n^{\eps/2})$. This also means that $S_G(v, i)$ contains $O(n^{\eps/2})$ nodes. Therefore, $M_v$ needs to send a total of $O(n^{\eps/2}) \times O(n^{\eps/2}) = O(n^\eps)$ words. 
A symmetric argument shows an $O(n^\eps)$ bound on the number of words $M_v$ receives. Since $O(n^\eps)$ words can be sent and received in each communication round, this communication can be completed in $O(1)$ rounds.
\end{proof}\fi{}
\end{proof}\fi{}
The biggest benefit from using this ``sample-and-gather'' simulation approach is for state-congested algorithms that sample a sparse subgraph and all activity occurs on this subgraph.
We formalize this sparse sampling property as follows.
\begin{definition}
Consider a state-congested algorithm $Alg$ that completes in $R$ rounds. For a parameter $\alpha \ge 2$, we say that $Alg$ is \textit{$\alpha$-sparse} if for all 
positive integers, $t$ and $\ell$ satisfying $t+\ell \le R$, for a length-$\ell$ phase of $Alg$ starting at round $t+1$ the following two properties hold.
\begin{itemize}
    \item[(a)] \textbf{Bounded activity level:} The activity level in the first round of the phase, $A_{t+1}$, satisfies the property: $A_{t+1} = O(\alpha^\ell \cdot \log n)$.
    \item[(b)] \textbf{Bounded growth of estimated activity level:} The estimated activity level $\tilde{A}_\tau$, $t+1 \le \tau \le t+\ell$, shows \textit{bounded growth}. Specifically, $\tilde{A}_{\tau+1} \le \alpha \tilde{A}_\tau$ for
    for all $t+1 \le \tau \le t+\ell-1$.
\end{itemize}
\end{definition}

\noindent
Together these properties require the activity level in each neighborhood to be low (Property (a)), but also that the \textit{estimated} activity  level of each node does not grow too fast in future rounds (Property (b)).  
When these two properties hold, the Lemma \ref{lemma:sampleAndGather} can
be applied inductively to obtain the following theorem. 
The fact that we use a single parameter $\alpha$ as an upper bound for both Properties (a) and (b) is just a matter of convenience and leads to an easy-to-state bound on number of rounds in this theorem.
\begin{theorem} \textbf{(Sample-and-Gather Theorem v1)}
\label{theorem:sampleAndGather}
Let $Alg$ be an $\alpha$-sparse state-congested algorithm that completes in $R$ rounds. 
Then $Alg$ can be simulated in the low-memory MPC model with  $\tilde{O}(m+n^{1+\eps})$ total memory, for $0 < \eps < 1$, in
$O\left(R \log\log n/\sqrt{\eps \cdot \log_\alpha n}\right)$
rounds.
\end{theorem}
\ifhideproofs{}\begin{proof}
Let $\ell = \lfloor \sqrt{\frac{\eps}{8} \cdot \log_{\alpha} n} \rfloor$.
Partition the $R$ rounds of $Alg$ into $\lceil R/\ell \rceil$ phases, where Phase $i$, $1 \le i < \lceil R/\ell \rceil$, consists of the $\ell$ rounds
$(i-1)\cdot \ell + 1, (i-1)\cdot \ell + 2, \ldots, i \cdot \ell$ and 
Phase $\lceil R/\ell \rceil$ consists of at most $\ell$ rounds 
$\left(\lceil R/\ell \rceil-1\right)\cdot \ell +1, \left(\lceil R/\ell \rceil-1\right)\cdot \ell + 2, \ldots, R$.

We now use the fact that $Alg$ is $\alpha$-sparse to show, via  series of inequalities,
that $\ell$ satisfies Inequality (\ref{equation:ellBound}). 
\begin{eqnarray*}
\left(\sum_{\tau = t+1}^{t+\ell} \tilde{A}_\tau \cdot \log n\right)^\ell & \le &  \left(\tilde{A}_{t+1} \cdot \log n \cdot \sum_{i = 0}^{\ell-1} \alpha^i \right)^\ell \qquad \text{\small (by Property (b) of being $\alpha$-sparse)}\\
 & \le & \left(A_{t+1} \cdot \log n \cdot \sum_{i = 0}^{\ell-1} \alpha^i \right)^\ell \qquad \text{\small (by $\tilde{A}_{t+1}=A_{t+1}$)}\\
 & \le & \left(\alpha^\ell \cdot \log^2 n \cdot \sum_{i = 0}^{\ell-1} \alpha^i \right)^\ell \qquad \text{\small (by Property (a) of being $\alpha$-sparse)}\\
 & = & \left(\alpha^\ell \cdot \log^2 n  \cdot \frac{\alpha^\ell - 1}{\alpha - 1} \right)^\ell \qquad \text{\small (by geometric series)}\\
 & \le & \alpha^{2\ell^2} \cdot (\log^2 n)^{\ell} \qquad\qquad \text{\small (by $\ell \ge 1$, $\alpha \ge 2$)}\\
 & \le & n^{\eps/4} \cdot n^{o(1)} \qquad\qquad\qquad \text{\small (by $\ell = \left\lfloor \sqrt{\frac{\eps}{8} \cdot \log_{\alpha} n} \right\rfloor$)}\\
 & \le & n^{\eps/2}.
\end{eqnarray*}
By using Lemma \ref{lemma:sampleAndGather}, this implies that each phase can be simulated in the MPC models with $O(n^\eps)$ memory per machine in $O(\log \ell) = O(\log\log n)$ rounds. 
Given that the $R$ rounds of $Alg$ are partitioned into $\lceil R/\ell \rceil$ phases,
we see that $Alg$ can be implemented in the MPC model with $O(n^\eps)$ memory per machine in $O(R \log\log n/\sqrt{\eps \log_\alpha n})$ rounds.
\end{proof}\fi{}

Theorem~\ref{theorem:sampleAndGather} provides a Simulation Theorem for the MPC model in which machines use $O(n^\eps)$ memory per machine.
However, the total memory used by MPC algorithms that result from this theorem is $\tilde{O}(m + n^{1+\eps})$. We now show that under fairly general circumstances, it is possible to obtain a Simulation Theorem yielding low-memory MPC algorithms that use only $\tilde{O}(m)$ total memory, while taking slightly more time.
\begin{definition}
A \textsc{Congest} algorithm $Alg$ is said to be \textit{degree-ordered} if it satisfies two properties.
\begin{itemize}
\item[(a)] The execution of $Alg$ can be partitioned into Stages $1, 2, \ldots$ such that in Stage $i$ the only active nodes are those whose degree is greater than
$\Delta^{1/2^i}$ and other nodes that are within $O(1)$ hops of these ``high degree'' nodes. 
\item[(b)] Let $R_i$ be the number of rounds in Stage $i$. Then $R_i \le R_{i-1}/2$. 
\end{itemize}
\end{definition}

\noindent
A lot of symmetry breaking algorithms are either inherently degree-ordered or can be made so with small modifications -- this can be seen in the applications of the Sample-and-Gather Theorems in Section \ref{sec:2rs}.
The fact that it is not just the ``high degree'' nodes, but even other nodes that are within $O(1)$ hops of high degree nodes that provides extra flexibility in this definition.
For algorithms that are degree-ordered, we can grow balls whose volume is at most 
the degree threshold for the current stage. This allows us to use a simple charging scheme to charge the sizes of the balls to the memory already allocated for node-neighborhoods.
This in turn yields the $\tilde{O}(m)$ total memory bound.
Property (b) holds for algorithms whose running time is dominated by $O(\log \Delta)$.
Given that the degree threshold in Property (a) falls as $\Delta^{1/2^i}$, the running time of each stage falls by a factor of 2.
\begin{lemma}
\label{lemma:mTotalMemory}
Suppose that $Alg$ is a state-congested, degree-ordered algorithm.
Consider a phase of $\ell-1$ rounds
$t+1, t+2, \ldots, t+\ell-1$ with a Stage $i$. If $\ell$ satisfies
\begin{equation}
\label{equation:ellBoundmTotal}
\left(\sum_{\tau = t+1}^{t+\ell} \tilde{A}_\tau \log n\right)^{\ell} \le \min\left\{n^{\eps/2}, \Delta^{1/2^i}\right\},
\end{equation}
then this phase can be simulated in $O(\log \ell)$ 
rounds in the low-memory MPC model with a total of $\tilde{O}(m)$ memory over all the machines.
\end{lemma}
\ifhideproofs{}\begin{proof}
The proof of this lemma is similar to the proof of Lemma \ref{lemma:sampleAndGather}. Here we point out the differences. First, in order to use the total memory more judiciously, we allow a machine to host multiple nodes, i.e., for distinct nodes $v$ and $v'$, the machines $M_v$ and $M_{v'}$ hosting these nodes may be identical.
Second, we only gather balls for ``high degree'' nodes, i.e., a machine $M_v$ gathers $S_G(v, \ell)$ iff $\deg(v) > \Delta^{1/2^i}$.  
With these modifications, 
Claims \ref{claim:SGsuffices}, \ref{claim:sizeBound}, and \ref{claim:gatheringTime} from the proof of Lemma \ref{lemma:sampleAndGather} hold, as before.
But, we need to additionally prove that a total of $\tilde{O}(m)$ memory is used by all the machines.

If a node $v$ is active in Stage $i$, then either (i) $\deg(v) > \Delta^{1/2^i}$ or (ii) $v$ has an active neighbor $u$ with $\deg(u) > \Delta^{1/2^i}$.
We deal with these two cases separately.
\begin{description}
    \item[Case (i):] By Claim \ref{claim:sizeBound}, the size of $S_G(v, \ell)$ is at most $\left(\sum_{\tau = t+1}^{t+\ell} \tilde{A}_\tau \log n\right)^{\ell}$. Therefore, by Inequality (\ref{equation:ellBoundmTotal}) the size of $S_G(v, \ell)$ is at most $\min\left\{n^{\eps/2}, \Delta^{1/2^i}\right\}$.
    Machine $M_v$ has at least $\min\left\{n^\eps, \deg(v)\right\}$ words of memory allocated to store the neighborhood of $v$.
    Since $\deg(v) > \Delta^{1/2^i}$ (because we are in Case (i)) the size of $S_G(v, \ell)$ can be charged to the memory allocated at machine $M_v$ to store neighbors of $v$. Thus the total size of all the gathered balls is at most 
    $(\sum_{v \in V} \deg(v)) \cdot \pl(n) = \tilde{O}(m)$.
    
    \item[Case (ii)] In this case, the ball $S_G(u, \ell)$ gathered by machine $M_u$ contains the ball $S_G(v, \ell-1)$. As a result, $M_u$ can figure out node $v$'s local state after $\ell-1$ rounds, $\sigma_{t+\ell-1}(v)$. 
\end{description}
\end{proof}\fi{}
Finally, if $Alg$ is a state-congested algorithm that is $\alpha$-sparse and degree-ordered, we obtain the following Simulation Theorem that guarantees an $\tilde{O}(m)$ total memory usage.
\begin{theorem}
\label{theorem:mTotalMemory}
\textbf{(Sample-and-Gather Theorem v2)}
Let $Alg$ be a state-congested, $\alpha$-sparse, degree-ordered algorithm that completes in $R$ rounds. 
Let $\alpha' = \alpha \cdot \log^2 n$.
Then $Alg$ can be simulated in the MPC model with $O(n^\eps)$ memory per machine, for $0 < \eps < 1$ and $\tilde{O}(m)$ total memory,  in
$O\left(R \log\log \Delta/\sqrt{\log_{\alpha'} \Delta}\right)$
rounds.
\end{theorem}
\ifhideproofs{}\begin{proof}
Consider a Stage $i$ for some positive integer $i$ and suppose that $Alg$ runs for $R_i$ rounds in this stage.
Let $\ell_i$ be a positive integer to be determined later.
Partition the $R_i$ rounds in Stage $i$ into $\lceil R_i/\ell_i \rceil$ 
phases, where Phases $j = 1, 2, \ldots, \lceil R_i/\ell_i \rceil-1$ consist of exactly $\ell_i$ rounds, whereas Phase $\lceil R_i/\ell_i \rceil$ consist of at most $\ell_i$ rounds.
We will now show that each Phase $j$ in Stage $i$ can be compressed into $O(\log \ell_i)$ low-memory MPC rounds.

We consider two cases depending on how $n^{\eps/2}$ compares with $\Delta^{1/2^i}$: (i) $\Delta^{1/2^i} \ge n^{\eps/2}$ and (ii) $\Delta^{1/2^i} < n^{\eps/2}$.
\begin{description}
\item[Case (i):] In this case, we set $\ell_i := \lfloor \sqrt{\frac{\eps}{8} \cdot \log_{\alpha} n} \rfloor$ as in the proof of Theorem \ref{theorem:sampleAndGather}.
Continuing as in the proof of this theorem, we conclude that the $\ell_i$ rounds in Phase $j$ in Stage $i$ can be compressed into $O(\log \ell_i)$ low-memory MPC rounds using $\tilde{O}(m)$ total memory.

\item[Case (ii):] Set $\ell_i := \lfloor \sqrt{\log_{\alpha'} \Delta^{1/2^{i+1}}/2} \rfloor$.
Suppose that $t+1$ is the index of the first round in Phase $j$ in Stage $i$.
Then, by calculations very similar to those in the proof of Theorem \ref{theorem:sampleAndGather}, we can show that
$$\left(\sum_{\tau = t+1}^{t+\ell_i} \tilde{A}_\tau \cdot \log n\right)^{\ell_i} \le \alpha^{2\ell_i^2} \cdot (\log^2 n)^{\ell_i} \le \Delta^{1/2^{i+1}} \cdot \left(\Delta^{1/2^i}\right)^{o(1)} \le \Delta^{1/2^i}.$$
Therefore, by Lemma \ref{lemma:mTotalMemory}, each Phase $j$ in Stage $i$ can be simulated in $O(\log \ell_i)$ rounds in the low-memory MPC model using $\tilde{O}(m)$ total memory. 
\end{description}

If it is the case that $\Delta^{1/2} \ge n^{\eps/2}$, then at least one of the initial stages of the algorithm will fall into Case (i). Theorem \ref{theorem:sampleAndGather} applies to each of these Case (i) stages and we see that all of these initial stages can be simulated in the low-memory MPC models with $\tilde{O}(m)$ total memory in
$$O\left(\frac{R \log\log n}{\sqrt{\eps \cdot \log_\alpha n}}\right) = O\left(\frac{R \log\log \Delta}{\sqrt{\eps \cdot \log_\alpha' \Delta}}\right)$$ rounds.
The upper bound above follows from the fact that $\Delta^{1/2} \ge n^{\eps/2}$, $\Delta \le n$, and $\alpha \le \alpha'$.

We now argue about the later stages, that fall into Case (ii) as follows.
Each Stage $i$ that is covered by Case (ii) can be simulated in 
$$O\left(\frac{R_i \log \ell_i}{\ell_i}\right) = O\left(\sqrt{2^{i+1}} \cdot \frac{R_i \log\log \Delta}{\sqrt{\log_{\alpha'} \Delta}}\right).$$
Therefore, the total running time of the Case (ii) stages of the simulated algorithm is
$$O\left(\sum_{i \ge 1} \sqrt{2^{i+1}} \cdot \frac{R_i \log\log \Delta}{\sqrt{\log_{\alpha'} \Delta}}\right) = O\left(\frac{\log\log \Delta}{\sqrt{\log_{\alpha'} \Delta}} \cdot \sum_{i \ge 1} \sqrt{2^{i+1}} \cdot R_i\right).$$
Finally, using Property (b) of a degree-ordered algorithm, i.e., the fact that $R_i \le R_1/2^{i-1}$, for all $i \ge 1$, we get that 
$$O\left(\frac{\log\log \Delta}{\sqrt{\log_{\alpha'} \Delta}} \cdot \sum_{i \ge 1} \sqrt{2^{i+1}} \cdot \frac{R_1}{2^{i-1}}\right) = 
O\left(\frac{\log\log \Delta}{\sqrt{\log_{\alpha'} \Delta}} \cdot \sum_{i \ge 1} \cdot \frac{R_1}{2^{(i-3)/2}}\right) = O\left(\frac{R \cdot \log\log \Delta}{\sqrt{\log_{\alpha'} \Delta}} \right).$$
Hence the total number of rounds is bounded as claimed in the theorem.
\end{proof}\fi{} 
\section{Fast 2-Ruling Set Algorithms}
Our 2-ruling set algorithms consist of 3 parts. In Part 1, we sparsify the input graph, in Part 2 we ``shatter'' the graph still active after Part 1, and in Part 3 we deterministically finish off the computation. 
Part 1 is a modification of \textsc{Sparsify}, a \congest\ model algorithm due to Kothapalli and Pemmaraju \cite{KothapalliP12}; Part 2 is a sparsified MIS algorithm, also in the \congest\ model, due to Ghaffari \cite{GhaffariSODA2016,Ghaffari2017}.
Our main contribution in this section is to show that these algorithms are state-congested, $\alpha$-sparse for small $\alpha$, and degree-ordered. 
As a result, we can apply the Sample-and-Gather Simulation Theorems (Theorems \ref{theorem:sampleAndGather} and \ref{theorem:mTotalMemory}) to these algorithms to obtain fast low-memory MPC algorithms.
Part 3 -- in which we finish off the computation -- is easy to directly implement in the MPC model.

\subsection{Simulating \textsc{Sparsify} in low-memory MPC}\label{sec:2rs}

Algorithm~\ref{alg:Sparsify} is a modified version of the \textsc{Sparsify} algorithm of Kothapalli and Pemmaraju~\cite{KothapalliP12}. The algorithm computes a ``sparse'' set of vertices \(U\) that dominates all the vertices in the graph (i.e.~\(\nbhd^{+}(U) = V\), see Lemma~\ref{lemma:sparsify}). In each iteration, ``high degree'' nodes and their neighbors are sampled and the sampled nodes are added to $U$. In successive iterations, the threshold for being a high degree node falls by a factor $f$ and the sampling probability grows by a factor $f$. The neighbors of the nodes that successfully join \(U\) are deactivated.

\RestyleAlgo{boxruled}
\begin{algorithm2e}\caption{\textsc{DegOrderedSparsify}\((G, f)\)\label{alg:Sparsify}}
\(U \leftarrow \emptyset\) \\
\(V_{0} \leftarrow V\) \tcp{Initially all nodes are active}
\For{\(i = 1\) to \(\lceil\log_{f}{\Delta}\rceil\)} {
    Let \(H_{i}\) be the nodes in \(V_{i-1}\) with degree at least \(\Delta/f^{i}\) in \(G[V_{i-1}]\)\\
    Each node in \(\nbhd^{+}(H_{i}) \cap V_{i-1}\) joins \(U_{i}\) with probability \(f^{i}\cdot c\ln n/\Delta\), where \(c\) is a fixed constant \\
    \(V_{i} \leftarrow V_{i-1} \setminus \nbhd^{+}(U_{i})\) \tcp{Nodes with at least one neighbor in \(U_{i}\) deactivate themselves}
    \(U \leftarrow U \cup U_{i}\) \\
}
\Return\ \(U\)
\end{algorithm2e}
\dos\ fits nicely within the framework of the Sample-and-Gather Simulation Theorems from Section \ref{section:SampleAndGather}.
The state of each vertex stays small throughout the algorithm (just $\texttt{ID}$ plus $O(1)$ bits), making \dos\ state-congested. 
The activity level in any iteration is bounded by $O(f \log n)$, because we show in Lemma \ref{lemma:sparsify} that in any neighborhood only $O(f \log n)$ vertices are sampled whp and only these sampled vertices need be active in that iteration. Furthermore, since the sampling probability grows by a factor $f$ in each iteration, the estimated neighborhood activity levels also grow by a factor $f$, as we consider future iterations of \dos.
As shown in Lemma \ref{lemma:sparsify}, this makes \dos\ $f$-sparse.
In the \textsc{Sparsify} algorithm \cite{KothapalliP12} \textit{all} nodes, independent of their degrees, sample themselves (as in Line 5). Here, in order to make \textsc{DegOrderedSparsify} degree-ordered, we make a small modification and permit only high degree nodes and their neighbors to sample themselves. As we show in Lemma \ref{lemma:sparsify}, the algorithm continues to behave as before, but is now degree-ordered.

\begin{lemma}
\label{lemma:sparsify}
Given a graph $G=(V,E)$ and a parameter $f > 3$, a subset $U\subseteq V$ can be computed in
$O(\log_f \Delta)$ rounds such that for every $v\in V$, $N^{+}(v) \cap U \neq \emptyset$, and for every $v\in U$, $\deg_{U}(v) \le 2cf\ln n$, with probability at least $1 - n^{-c+2}$.
\end{lemma}
\ifhideproofs{}\begin{proof}
Consider an execution of $\textsc{DegOrderedSparsify}(G,f)$. Assume, inductively, that just before the \(i^{t h}\) iteration the maximum degrees in the graphs induced by $V_{i-1}$  and $U_{i-1}$ are at most $\Delta/f^{i-1}$ and $f\cdot 2c\ln n$ respectively. These bounds hold trivially when $i=1$.

Each $v \in N^{+}(H_{i})$ is included in $U_{i}$ independently with probability $c\ln n f^{i}/\Delta$, so the probability that a $v \in V_{i-1}$ with $\deg_{V_{i-1}}(v) > \Delta/f^i$ is not in $N^{+}(U_{i})$ is less than ${(1-c\ln n f^i/\Delta)}^{\Delta/f^i} < n^{-c}$.

Furthermore, if $v\in U_{i}$,
\[
\mathbf{E}[\deg_{U_i}(v)] = \deg_{V_{i-1}}(v) \cdot f^i \cdot c\ln n/\Delta \le (\Delta/f^{i-1}) \cdot f^i \cdot c\ln n/\Delta\le cf\ln n.
\]

Here, the second inequality follows from the inductive hypothesis. Using the standard Chernoff bound, the probability that $\deg_{U_i}(v) \ge 2cf\ln n$ is at most $e^{(-fc\ln n/3)} < n^{-c}$ (because \(f \ge 3\)).
Note that since $v$ and its neighborhood are permanently removed from consideration,
it never acquires new neighbors in $U$, so $\deg_{U_{i}}(v) = \deg_{U}(v)$.
Thus, by the union bound, the induction hypothesis does not hold for the next iteration with probability at most $n^{-c + 1}$. And since there are at most \(\log_{f} \Delta \le n\) iterations, the probability that there exists a node \(v \in U\) with \(\deg_{U}(v) \ge 2cf\ln n\) is at most $n^{-c + 2}$ by the union bound.
\end{proof}\fi

It is easy to see that the algorithm \(\textsc{DegOrderedSparsify}(G,f)\) can be implemented in the \textsc{Congest} model in $O(\log_f \Delta)$ rounds because each iteration of the \textbf{for}-loop takes $O(1)$ rounds in \congest. Furthermore, since each node can update its state by simply knowing if it or a neighbor has joined set $U_i$, the update function at each node is separable (see Definition \ref{def:separable}). Therefore, \(\textsc{DegOrderedSparsify}(G,f)\) can be faithfully simulated in the low-memory MPC model in \(O(\eps^{-1} \log_f \Delta)\) rounds.

\ifhideproofs{}
\begin{theorem}\label{theorem:sparsifySimpleSimulation}
The algorithm \(\textsc{DegOrderedSparsify}(G,f)\) can be simulated in \(O(\eps^{-1} \log_f \Delta)\) rounds whp in the low-memory MPC model with \(\tilde{O}(m)\) total memory.
\end{theorem}
\begin{proof}
  There are \(O(\log_{f} \Delta)\) iterations of the for loop. In the \(i^{th}\) iteration, the nodes need to know their degree in the acitve subgraph, and whether it has a neighbor in \(U_{i}\) or not. These operations can be encoded as separable functions, and therefore can be performed in \(O(\eps^{-1})\) rounds using Lemma~\ref{lem:generalmaxload}. The rest of the steps for a node can be done locally at the host machine and only requires this information. This means that the host machine of each node knows whether or not it has joined \(U_{i}\) or \(V_{i}\) or neither. Therefore, each iteration requires \(O(\eps^{-1})\) rounds which proves the lemma.
\end{proof}\fi{}
We now show that \textsc{DegOrderedSparsify} has the three properties needed for round compression via our Simulation Theorems and this leads to a substantial speedup. 

\begin{lemma}\label{theorem:sparsifyproperties}
The algorithm \(\textsc{DegOrderedSparsify}(G,f)\) is a state-congested, $f$-sparse, degree-ordered algorithm.
\end{lemma}
\ifhideproofs{}\begin{proof}
  We first show that \textsc{DegOrderedSparsify} (Algorithm~\ref{alg:Sparsify}) can be implemented in the \textsc{Congest} model in a state congested fashion. We have aleady shown that it is easy to implement \(\textsc{DegOrderedSparsify}(G,f)\) in the \textsc{Congest} model using a constant number of rounds per iteration. Therefore, it suffices to ensure that the state of each node \(v\) can always be represented using \(O(\deg(v) \log n)\) bits. Initially, the state of each node consists of its ID and the ID's of all its neighbors, and the maximum degree and number of vertices in the graph. This information can be stored using \(O(\deg(v) \log n)\) bits at each node \(v\). In each iteration \(i\) of the for loop, the nodes just need to keep track of whether they are in \(V_{i-1}\) or not, and if \(v\) joins \(U_{i}\) or not. This only adds a constant number of bits to the state of each node.

  Consider a round \(t\) corresponding to iteration \(i\) in the state-congested implementation of \textsc{DegOrderedSparsify}. In iteration \(i\), a node \(v\) becomes a sending node with activation probability \(p_t(v) \le f^{i}\ln n/\Delta\). Moreover, only nodes with degree at least \(\Delta_{i} = \Delta/f^{i}\) and their neighbors sample themselves into joining \(U_{i}\).  The graph induced by the active nodes \(G[V_{i-1}]\) has maximum degree \(\le \Delta_{i-1} = \Delta/f^{i-1}\) with high probability due to Lemma~\ref{lemma:sparsify}. Now for a node \(v\), we have \(A_t(v) := \sum_{u \in Nbr(v)} p_t(u) \le f \ln n\). Therefore, the maximum over all fully-active nodes is \(A_{t} = f \ln n\).

  Further, we can find an appropriate number of rounds $r_i$ that satisfies $\Delta^{1/2^{i-1}}/f^{r_i} \leq \Delta^{1/2^i}$, $i=1,2,\dots, O(\log \log \Delta)$. One can think of these $r_i$ consecutive rounds of Algorithm \textsc{DegOrderedSparsify} as a logical stage in which nodes that participate in that stage all have a degree at least $\Delta^{1/2^i}$. Since the maximum degree of a node across these logical stages falls by a factor of $\Delta^{1/2^{i}}$, the number of rounds needed across two consecutive logical stages falls by a factor of 2. In this view, Algorithm \textsc{DegOrderedSparsify} satisfies both the conditions required of a degree-ordered algorithm.

\end{proof}\fi{}

Using Theorem~\ref{theorem:sampleAndGather} and Theorem~\ref{theorem:mTotalMemory}, we obtain the following theorem.
\begin{theorem}\label{theorem:sparsifyMPC}
The algorithm \(\textsc{DegOrderedSparsify}(G,f)\) can be implemented in the low-memory MPC model in (i) $O\left(\eps^{-1/2}\frac{\log_f \Delta}{\sqrt{\log_f n}}\log\log n\right)$
rounds whp using $\tilde{O}(m + n^{1+\eps})$ total memory and (ii) $\left(\frac{\log_f \Delta}{\sqrt{\log_f n}}\log\log \Delta\right)$
rounds whp using $\tilde{O}(m)$ total memory.
\end{theorem}
\ifhideproofs{}\begin{proof}
We appeal to Theorem~\ref{theorem:sampleAndGather} with  \(\alpha = f\) and \(R = \log_{f} \Delta\) to obtain (i).
For part (ii), we appeal to Theorem~\ref{theorem:mTotalMemory} with $\alpha' = \alpha \cdot \log^2 n$, $R = \log_f \Delta$, and $\alpha = f$. These values for the parameters allow Algorithm  \(\textsc{DegOrderedSparsify}(G,f)\) to be simulated in the low-memory MPC model with $\tilde{O}(m)$ total memory in $O\left ( \frac{\log_f \Delta \cdot \log\log \Delta}{\sqrt{\log_{\alpha'} n} }\right )$ rounds which simplifies to $O(\sqrt{\log_f \Delta} \cdot \log\log \Delta)$.
\end{proof}\fi{}

\subsection{Simulating Sparsified Graph Shattering in low-memory MPC}
Distributed graph shattering has become an important algorithmic technique for symmetry breaking problems~\cite{BEPS16,Ghaffari2017,Ghaffari2019}. In this section, we use a sparsified graph shattering algorithm due to Ghaffari \cite{Ghaffari2017} to process the graph $G[U]$ returned by \dos. The output of the shattering algorithm consists of an independent set $I \subseteq U$ such that the graph induced by the remaining set of vertices $S = U \setminus \nbhd^+(I)$

contains only small connected components.

Ghaffari's sparsified shattering algorithm \cite{Ghaffari2017} is shown in Algorithm \ref{alg:shatter}.
At the start of each round $t$, each node $v$ has a \textit{desire-level}
$p_t(v)$ for joining the independent set $I$, and initially this is set to
$p_1(v) = 1/2$. The independent set $I$ is also initialized to the empty set.
The algorithm runs in phases, with each phase having $\ell := \sqrt{\delta \log n}/10$ rounds for a small constant $\delta$. 

Several aspects of the algorithm make it nicely fit the Sample-and-Gather framework from Section \ref{section:SampleAndGather}. We now point these out.
(i) The desire-level $p_\tau(u)$ for $t+1 \le \tau \le t+\ell$ can be viewed the probability of sampling $u$; after the initial communication amongst neighbors (Line 1), only sampled nodes send messages (beeps) and all other nodes remain silent.
(ii) The quantity $d_{t+1}(u)$ is identical to the activity level $A_{t+1}(u)$ in $u$'s neighborhood, defined in Section \ref{section:SampleAndGather}.
(iii) Nodes with a high activity level, i.e., $d_{t+1}(u) \ge 2^{\sqrt{\log n}/5}$ (aka super-heavy nodes), are send-only nodes and are therefore excluded in the definition of $A_{t+1}$. As a result $A_{t+1} \le 2^{\sqrt{\log n}/5}$.
(iv) In each iteration in a phase, the sampling probability grows by a factor of at most 2 (Line 8). This implies that the estimated activity levels grow by a factor of 2 in future rounds.

\RestyleAlgo{boxruled}
\begin{algorithm2e}\caption{\textsc{Shatter}($G$): (one phase, starting at iteration $t+1$) \label{alg:shatter}}
    Each node $u$ sends its current desire-level $p_{t+1}(u)$ to all its neighbors \\
    Each node $u$ computes $d_{t+1}(u) = \sum_{v \in \nbhd(u)} p_{t+1}(v)$ \\
    If node $u$ has $d_{t+1}(u) \geq 2^{\sqrt{\log n}/5}$ then $u$ is called a \textit{super-heavy} node \\
    $\ell = \sqrt{\delta \log n}/10$  \tcp*{$\delta$ is a small constant}
    \For{$\tau = t+1, t+2, \ldots, t+\ell$ iterations } {
    \tcp{Round 1}
    Each node $u$  beeps with probability $p_\tau(u)$ and remains silent otherwise.\\
    Node $u$ is added to $I$ if it is not super-heavy, it beeps, and none of its neighbors beep \\
    Node $u$ sets $p_{\tau+1}(u)$ as follows:
    \[p_{\tau+1}(u) = \begin{cases}
        p_\tau(u)/2 & \mbox{ if $u$ is super-heavy, or a neighbor of $u$ beeps} \\
        \min\{1/2, 2\cdot p_\tau(u)\} & \mbox{ otherwise }
    \end{cases} 
    \] 
    \tcp{Round 2}
    Node $u$ beeps if it joins $I$ in this iteration. \\ 
    Neighbors of node $u$ that are not in $I$ become inactive on hearing the beep from $u$
    }
\end{algorithm2e}
The first four steps of Algorithm~\ref{alg:shatter} do not fit into the Sample-and-Gather framework since each node needs to send its \(p_{t+1}\) value to its neighbors. But the nodes are computing $d_{t+1}(u) = \sum_{v \in \nbhd(u)} p_{t+1}(v)$ which is a separable function (sum). Therefore, we can implement the first two steps in \(O(1/\eps)\) rounds using Lemma~\ref{lem:generalmaxload}, and use the Sample-and-Gather framework to simulate the \textbf{for}-loop of the algorithm.
These observations are formalized in the lemma below to show that \textsc{Shatter} is 2-sparse. Additionally, the lemma shows that the algorithm is state-congested.

\begin{lemma}\label{lemma:misprops}
  Algorithm~\ref{alg:shatter} is a state-congested algorithm whose \textbf{for}-loop is $2$-sparse.
\end{lemma}
\ifhideproofs{}\begin{proof}
We first note that Algorithm~\ref{alg:shatter} can be implemented in the \textsc{Congest} model using \(O(1)\) rounds per iteration. The state of each node is also very small, it consists of the \texttt{ID} and \(p_{\tau}(\cdot)\) value plus $O(1)$ bits for additional bookkeeping. The \(p_{\tau}(\cdot)\) values require \(O(\log n)\) bits of precision since we won't run the algorithm for more than \(O(\log n)\) rounds.

From the discussion above, the initial activity level around a node at the beginning of the phase $A_{t} = O(2^{\sqrt{\log n}} \log n)$ and the activity level at each node can increase by a factor of at most \(2\) in each iteration. Therefore, the \textbf{for}-loop of Algorithm~\ref{alg:shatter} becomes $2$-sparse.
\end{proof}\fi{}

A total of $O(\log \Delta/\sqrt{\log n})$ repeated applications of \textsc{Shatter} (i.e.~a total of \(O(\log \Delta)\) iterations) suffice to shatter the graph into small-sized components~\cite{Ghaffari2017,Ghaffari2018}.
\ifhideproofs{}Specifically, the following theorem is proved.

\begin{theorem}\label{thm:postShatter}
Suppose that we execute a total $O(\log\Delta)$ iterations of \textsc{Shatter} (partitioned into phases, each with $\sqrt{\delta \log n}/10$ iterations each).
The set $I$ of vertices is independent.
Furthermore, the set $S = U \setminus \nbhd^+(I)$ of nodes that remain in the graph satisfy the following three properties whp: 
(i) Each connected component of the graph induced by $S$ has $O(\Delta^4 \cdot \log_{\Delta}  n )$ nodes, (ii) $|S| \leq n/\Delta^{10}$, and (iii) If $\Delta  > n^{\alpha/4}$ then the set $S$ is empty.
\end{theorem}
\fi{}

Using Lemma \ref{lemma:misprops} and Theorem \ref{theorem:sampleAndGather}, we obtain the following lemma that shows that \textsc{Shatter} can be simulated efficiently in the low-memory MPC model.

\begin{lemma}\label{lem:morememoryShattering}
 We can simulate a total \(O(\log \Delta)\) iterations of Algorithm \textsc{Shatter} in the low-memory MPC model with $\tilde{O}(m+n^{1+\eps})$ total memory in
$O\left(\frac{\log \Delta \cdot \log\log n}{\eps\sqrt{ \log n}} \right)$ rounds whp.
\end{lemma} 
\ifhideproofs{}\begin{proof}
We partition the \(O(\log \Delta)\) iterations into phases, each with $\sqrt{\delta \log n}/10$ iterations each. In each phase we run Algorithm~\ref{alg:shatter}.

As stated earlier, the first four steps of Algorithm~\ref{alg:shatter} do not fit into the Sample-and-Gather framework since each node needs to send its \(p_{t+1}\) value to its neighbors. But the nodes are computing $d_{t+1}(u) = \sum_{v \in \nbhd(u)} p_{t+1}(v)$ which is a separable function (sum). Therefore, we can implement the first two steps in \(O(1/\eps)\) rounds using Lemma~\ref{lem:generalmaxload}.

Based on Lemma~\ref{lemma:misprops}, the $\sqrt{\delta \log n}/10$  the iterations of Algorithm \textsc{Shatter} form a \(2\)-sparse state congested algorithm. So, we can appeal to Theorem~\ref{theorem:sampleAndGather} (with \(\alpha=2\) and \(R = \sqrt{\delta \log n}/10\) and simulate the \textbf{for} loop of Algorithm \textsc{Shatter} in $O\left(\log\log n/\sqrt{\eps}\right)$ rounds in the low-memory MPC model with $\tilde{O}(m+n^{1+\eps})$ total memory. 

So a single phase can be simulated in $O\left(1/\eps + \log\log n/\sqrt{\eps}\right)$ rounds. Over all the phases, we get the number of rounds for simulation is $O\left(\frac{\log \Delta \cdot  \log \log n}{\eps \sqrt{\log n}}\right)$.
\end{proof}\fi{}

Ghaffari and Uitto \cite{Ghaffari2018} present a variant of Algorithm \textsc{Shatter} and show (in Theorem 3.7) that this variant can be simulated in $O(\sqrt{\log \Delta} \cdot \log\log \Delta)$ rounds in the low-memory MPC model, while using only $\tilde{O}(m)$ total memory.
While they describe their MPC implementation from scratch, this MPC implementation can also be obtained by applying our Sample-and-Gather Theorem (specifically, Theorem \ref{theorem:mTotalMemory}).
It can be shown that this variant is state-congested and has the same sparsity property as Algorithm \textsc{Shatter}, i.e., it is 2-sparse.
Furthermore, it can also be made degree-ordered by simply processing nodes in degree buckets $(\Delta^{1/2^i}, \Delta^{1/2^{i-1}}]$, in the order $i = 1, 2, \ldots, O(\log \log \Delta)$.

\begin{lemma} [Ghaffari-Uitto~\cite{Ghaffari2018}]\label{lem:mShattering}
There is a variant of Algorithm \textsc{Shatter} can be simulated  in the low-memory MPC model with $\tilde{O}(m)$ total memory in
$O(\sqrt{\log \Delta} \cdot \log\log \Delta)$ rounds whp.
\end{lemma} 

\subsection{Finishing off the 2-ruling set computation}
After applying \dos\ to the input graph $G = (V, E)$ and then \textsc{Shatter} to the subgraph $G[U]$, induced by the subset $U \subseteq V$ output by \dos, we are left with a number of small-sized components\ifhideproofs{}, as shown in Theorem \ref{thm:postShatter}\fi{}.
Ghaffari and Uitto~\cite[Theorem 3.7]{Ghaffari2018} show that given the properties that the remaining graph has after \textsc{Shatter}, it is possible to simply (and deterministically) gather each component at a machine and find an MIS of the component locally in $O(\sqrt{\log\log n})$ rounds in the low-memory MPC model using $\tilde{O}(m)$ memory.
Applying this ``finishing off'' computation completes our 2-ruling set algorithm.
The output of the algorithm is the union of the independent set output by \textsc{Shatter} and the independent set output by the ``finishing off'' computation.

\begin{theorem}\label{theorem:2rslow}
A \(2\)-ruling set can be computed whp in the low-memory MPC model in \\
(i) $O(\eps^{-1}(\log \Delta)^{1/6} \log\log n)$ rounds using \(\tilde{O}(m + n^{1+\eps})\) total memory and in \\ (ii) \(O((\log \Delta)^{1/4} \log\log \Delta + \sqrt{\log\log n} \log\log \Delta)\) rounds using \(\tilde{O}(m)\) total memory.
\end{theorem}
\ifhideproofs{}\begin{proof}
We first prove Part (i).
In order to get a \(2\)-ruling set algorithm, we first run \dos\ with a parameter \(f\) to get a set \(U\). We then run \textsc{Shatter} on the induced subgraph \(G[U]\) to get an independent set $I$. Finally, we run the ``finishing off'' computation on the subgraph $G[U \setminus \nbhd^+(I)]$ to get an independent set $I'$. Since $U$ dominates all vertices in $V$ (by Lemma \ref{lemma:sparsify}) and $I \cup I'$ is an MIS of $G[U]$, we see that $I \cup I'$ is a 2-ruling set of $G$.

We now bound the running time of the algorithm as follows.
Using the running time of \dos\ from Theorem~\ref{theorem:sparsifyMPC} part (i), the running time of \textsc{Shatter} from Lemma~\ref{lem:morememoryShattering}, and the fact that the ``finishing off'' computation runs in $O(\sqrt{{\log\log n})}$ rounds, we get a total running time of
  \[O\left(\eps^{-1}\left(\frac{\log \Delta}{\sqrt{\log f \log n}} + \frac{\log( f \log n)}{\sqrt{\log n}} \right)\log\log n + \sqrt{\log\log n}\right).\]
The $\log(f \log n)$ term in numerator of the second term above is due to the fact that the maximum degree in $G[U]$ is bounded above by $O(f \log n)$, as shown in Lemma \ref{lemma:sparsify}. This expression is minimized at $f = 2^{(\log \Delta)^{2/3}}$. Plugging this value and simplifying yields a running time of 
$$O\left(\eps^{-1} \frac{(\log \Delta)^{2/3}}{\sqrt{\log n}} \log \log n\right) = O\left(\eps^{-1}(\log \Delta)^{1/6} \log \log n\right).$$

We now prove Part (ii).
The correctness of our 3-part algorithm has already been established. To bound the running time, we use the running time bound for \textsc{DegOrderedSparsify} from Theorem~\ref{theorem:sparsifyMPC} part (ii), the running time bound on \textsc{Shatter} from Lemma \ref{lem:mShattering}, and the fact that the ``finishing off'' computation runs in $O(\sqrt{\log\log n})$ rounds to
get a running time of
\[O\left(\sqrt{\frac{\log \Delta}{\log f}}  \log\log \Delta + \sqrt{\log (f \log n)} \cdot \log\log \Delta +  \sqrt{\log\log n} \right).\]
This expression is minimized at $f = 2^{(\log \Delta)^{1/2}}$. Plugging this value of $f$ and simplifying yields a total running time of \(O((\log \Delta)^{1/4} \log\log \Delta + \sqrt{\log\log n} \log\log \Delta)\).
\end{proof}\fi{}

\noindent
\textbf{Remark:} We note that by just running \textsc{Shatter} on an input graph followed by the ``finishing off'' computation, we get an MIS of the input graph.
So our approach yields MIS algorithms in the low-memory MPC model via the Sample-and-Gather Simulation Theorems.

\begin{theorem}\label{thm:finishMIS}
An MIS of a graph $G$ can be found in the low-memory MPC model in: (i) $O\left(\frac{\log\Delta \cdot  \log \log n}{\eps\sqrt{\log n}} + \sqrt{\log\log n}\right)$ rounds whp using $\tilde{O}(m+n^{1+\eps})$ total memory and (ii) $O(\sqrt{\log \Delta} \log \log \Delta + \sqrt{\log\log n} )$ rounds whp using $\tilde{O}(m)$ total memory.
\end{theorem} 
As far as we know, the MIS result for the $\tilde{O}(m+n^{1+\epsilon})$ total memory setting is new, but the result for the $\tilde{O}(m)$ total memory setting simply recovers the result from \cite{Ghaffari2018}.

\section{Fast $\beta$-ruling Set Algorithms}
We now extend the \(2\)-ruling set low-memory MPC algorithm in the previous section to obtain a $\beta$-ruling set low-memory MPC algorithm for any integer $\beta \ge 2$. The overall idea is to repeatedly use an \textsc{DegOrderedSparsify}, as in~\cite{BishtKP13}. We start by running a low-memory MPC implementation of \textsc{DegOrderedSparsify} with a parameter $f_1$; this call returns a  set of nodes \(S_{1}\). Once this phase ends, the remaining graph $G[S_1]$ has degree at most $O(f_1 \cdot \log n)$, by Lemma~\ref{lemma:sparsify}.
We then run \textsc{DegOrderedSparsify} on the graph $G[S_1]$ with a parameter $f_2$ and this yields a  set of nodes $S_2$.
This process continues for $\beta-1$ phases at the end of which the graph $G[S_{\beta-1}]$ has a maximum degree $O(f_{\beta-1} \cdot \log n)$. We now proceed to run a low-memory MPC implementation of an MIS algorithm on $G[S_{\beta-1}]$. \ifhideproofs{}This returns a  set of nodes \(C\), that turns out to be a $\beta$-ruling set of the input graph $G$.\fi{}
\ifhideproofs{} A pseudo-code of this algorithm is shown below as Algorithm~\ref{alg:betars}.
At this point, we leave the parameters $f_1, f_2, \ldots, f_{\beta-1}$
unspecified. Later, we instantiate values for these parameters so as to minimize the overall running time in two different settings: (i) when the total memory is bounded by $\tilde{O}(m+n^{1+\eps})$ and (ii) when the total memory is bounded by $\tilde{O}(m)$.  

\RestyleAlgo{boxruled}
\begin{algorithm2e}\caption{\(\beta\)-\textsc{RulingSet}\((G, f_1, f_2, \dots, f_{\beta-1})\)\label{alg:betars}}
\(S_0 \leftarrow V\) \\
\For{\(i = 1\) to \(\beta-1\)} {
    \(S_i \leftarrow \textsc{DegOrderedSparsify}(G[S_{i-1}], f_i)\) \\
}
Nodes in \(S_{\beta-1}\) compute an MIS \(C\) of \(G[S_{\beta-1}]\)~\label{alg:betars:MIS}\\
\Return\ \(C\)
\end{algorithm2e}
\fi{}

The correctness of \ifhideproofs{} Algorithm~\ref{alg:betars}\fi{} can be noted from Lemma~\ref{lemma:sparsify}. The set \(S_{i}\) covers all the nodes in \(S_{i-1}\) which means that all the nodes in \(S_{0} = V\) are at most \(\beta-1\) hops away from the nodes in \(S_{\beta-1}\). Therefore all the nodes of \(V\) are at most \(\beta\) hops away from the MIS \(C\) of \(G[S_{\beta-1}]\). This means that the set \(C\) that \ifhideproofs{} Algorithm~\ref{alg:betars}\fi{} returns is a \(\beta\)-ruling set of \(G\). In the following, we analyze the round complexity of \ifhideproofs{}Algorithm~\ref{alg:betars}\fi{} in the low-memory MPC model.

\begin{lemma}\label{lemma:betarsSuperLinearMPC}
   Let \(f_{0} = \Delta\). \ifhideproofs{} Algorithm~\ref{alg:betars}\fi{} runs in
   \begin{equation}
   O\left(\eps^{-1}\left(\sum_{i=1}^{\beta-1} \frac{\log (f_{i-1} \log n)}{\sqrt{\log f_{i} \cdot \log n}} + \frac{\log (f_{\beta-1} \log n)}{\sqrt{\log n}} \right)\log\log n \right) \label{equation:betaRuling1}
   \end{equation}
   rounds whp in the low-memory MPC model with \(\tilde{O}(m + n^{1+\eps})\) total memory.
\end{lemma}
\ifhideproofs{}\begin{proof}
  Consider an arbitrary iteration \(i\), where \(1 \le i \le \beta-1\). The set \(S_{i-1}\) is the output of running \textsc{DegOrderedSparsify} with parameter \(f_{i-1}\). Therefore, by Lemma~\ref{lemma:sparsify}, we can say that the subgraph \(G[S_{i-1}]\) has maximum degree at most \(O(f_{i-1}\log n)\) whp. Therefore the running time of \(\textsc{DegOrderedSparsify}(G[S_{i-1}], f_i)\) will be
  \[O\left(\eps^{-1/2}\frac{\log(f_{i-1}\log n)}{\sqrt{\log f_{i} \cdot \log n}}\log\log n\right)\]
  rounds in the low-memory MPC model by Theorem~\ref{theorem:sparsifyMPC} part (i) with \(\tilde{O}(m + n^{1+\eps})\) total memory.

After the $\beta-1$ calls to \textsc{DegOrderedSparsify} are completed, in Line~\ref{alg:betars:MIS} of Algorithm~\ref{alg:betars}, we call the MIS algorithm referred to in Theorem~\ref{thm:finishMIS} part (i).
Since the max.~degree of the graph that is input to this call is $O(f_{\beta-1} \cdot \log n)$, this call runs in $O\left(\eps^{-1}\left(\frac{\log( f_{\beta-1}\log n)}{\sqrt{\log n}}\right)  \cdot \log\log n\right)$ rounds for obtaining an MIS of \(G[S_{\beta-1}]\). The lemma follows.
\end{proof}\fi{}

\begin{lemma}\label{lemma:betarsLinearMPC}
   Let \(f_{0} = \Delta\). \ifhideproofs{} Algorithm~\ref{alg:betars}\fi{} runs in
   \begin{equation}
   O\left(\left(\sum_{i=1}^{\beta-1} \sqrt{\frac{\log (f_{i-1}\log n)}{\log f_{i}}} + \sqrt{\log( f_{\beta-1}\log n)} \right)\log\log \Delta  + \sqrt{\log \log n}\right) \label{equation:betaRuling2}
   \end{equation}
   rounds whp in the low-memory MPC model with \(\tilde{O}(m)\) total memory.
\end{lemma}
\ifhideproofs{}\begin{proof}
  Consider an arbitrary iteration \(i\), where \(1 \le i \le \beta-1\). The set \(S_{i-1}\) is the output of running \textsc{DegOrderedSparsify} with parameter \(f_{i-1}\). Therefore, by Lemma~\ref{lemma:sparsify}, we can say that the subgraph \(G[S_{i-1}]\) has maximum degree at most \(O(f_{i-1}\log n)\) whp. Therefore the running time of \(\textsc{DegOrderedSparsify}(G[S_{i-1}], f_i)\) will be
  \[O\left(\sqrt{\frac{\log (f_{i-1}\log n)}{\log f_{i}}} \log\log \Delta\right)\]
  rounds in the low-memory MPC model by Theorem~\ref{theorem:sparsifyMPC} part (ii) with \(\tilde{O}(m)\) total memory.

After the $\beta-1$ calls to \textsc{DegOrderedSparsify} are completed, in Line~\ref{alg:betars:MIS} of Algorithm~\ref{alg:betars}, we call the MIS algorithm referred to in Theorem~\ref{thm:finishMIS} part (ii).
Since the maximum degree of the graph that is input to this call is $O(f_{\beta-1} \cdot \log n)$ whp, this call runs in \(O\left(\left(\sqrt{\log( f_{\beta-1}\log n)}\right)  \cdot \log\log \Delta + \sqrt{\log \log n}\right)\) rounds for obtaining an MIS of \(G[S_{\beta-1}]\). The lemma follows.
\end{proof}\fi{}

We now instantiate the parameters $f_1, f_2, \ldots, f_{\beta-1}$ 
so as to minimize the running times in Lemmas~\ref{lemma:betarsSuperLinearMPC} and~\ref{lemma:betarsLinearMPC}.
This leads to the following corollaries.

\begin{theorem}\label{thm:betarsMPC}
A $\beta$-ruling set of a graph $G$ can be found whp in the low-memory MPC model in (i) $O\left(\eps^{-1}\beta \cdot \log^{1/(2^{\beta+1}-2)} \Delta \cdot  \log\log n\right)$ rounds with \(\tilde{O}(m + n^{1+\eps})\) total memory and in (ii) \(O\left(\beta \cdot \left(\log^{1/2\beta} \Delta \cdot \log\log \Delta + \sqrt{\log\log n}\right) \cdot \log\log \Delta\right)\)
rounds with \(\tilde{O}(m)\) total memory.
\end{theorem}
\ifhideproofs{}\begin{proof}
For part (i), we set $f_i := 2^{\log^{\delta_i} \Delta}$ for all $i$, $1 \le i \le \beta -1$. We then set
\[
\delta_{i-1} =
\begin{cases}
\frac{1}{2} + \frac{1}{2^{\beta+1}-2} & \text{if } i = \beta\\
\frac{1}{2} + \frac{1} {2^{\beta+1}-2} + \frac{\delta_i}{2} & \text{if } 1 < i < \beta.\\
\end{cases}
\]
With this setting of the parameters, the term in (\ref{equation:betaRuling1}) after the summation evaluates to
\begin{align*}
 O&\left(\eps^{-1}\left(\frac{\sqrt{\log \Delta} \cdot \log^{1/(2^{\beta+1}-2)} \Delta + \log\log n}{\sqrt{\log n}}\right) \cdot \log\log n\right) \\
 &= O\left(\eps^{-1} \log^{1/(2^{\beta+1}-2)} \Delta \cdot \log\log n\right).
\end{align*}
The term in the summation in (\ref{equation:betaRuling1}) indexed by $i$ for $1 \le i \le \beta-1$ evaluates to
\begin{align*}
O&\left(\eps^{-1}\left(\frac{\sqrt{\log \Delta} \cdot \log^{1/(2^{\beta+1}-2)} \Delta
\cdot \log^{\delta_i/2} \Delta
+ \log\log n}{\log^{\delta_i/2} \Delta \cdot \sqrt{\log n}}\right) \cdot \log\log n\right) \\ 
&= O\left(\eps^{-1} \log^{1/(2^{\beta+1}-2)} \Delta \cdot \log\log n\right).
\end{align*}
This yields the claimed running time because each of the $\beta$ terms in the expression for the running time in (\ref{equation:betaRuling1}) is equal to $O\left(\eps^{-1} \log^{1/(2^{\beta+1}-2)} \Delta \cdot \log\log n\right)$.

For part (ii), we set $f_i := 2^{\log ^{1-\frac{i}{\beta}} \Delta}$ for all $i$, $1 \le i \le \beta-1$.
With this setting of the parameters, the term in (\ref{equation:betaRuling2}) after the summation evaluates to
\begin{align*}
O&\left(\sqrt{\log^{1/\beta} \Delta + \log\log n} \cdot \log\log \Delta\right) \\
&= O\left(\left(\log^{1/2\beta} \Delta \cdot \log\log \Delta + \sqrt{\log\log n}\right) \cdot \log\log \Delta\right).
\end{align*}

The term in the summation in (\ref{equation:betaRuling2}) indexed by $i$ for $1 \le i \le \beta-1$ evaluates to
\begin{align*}
O&\left(\sqrt{\frac{\log^{1-(i-1)/\beta} \Delta + \log\log n}{\log^{1-i/\beta} \Delta}}\cdot \log\log \Delta\right) \\
&= O\left(\left(\log^{1/2\beta} \Delta \cdot \log\log \Delta + \sqrt{\log\log n}\right) \cdot \log\log \Delta\right).
\end{align*}

This yields the claimed running time because each of the $\beta$ terms in the expression for the running time in (\ref{equation:betaRuling2}) is equal to $O\left(\left(\log^{1/2\beta} \Delta \cdot \log\log \Delta + \sqrt{\log\log n}\right) \cdot \log\log \Delta\right)$.
\end{proof}\fi{}

\subsection{$\beta$-ruling sets in $O(\pll(n))$ rounds}

As mentioned in the Introduction, this research is partly motivated by the question of whether ruling set problems can be solved in the low-memory MPC model in $O(\pll(n))$ rounds.
Using our results we identify two interesting circumstances under which $\beta$-ruling sets can be computed in the low-memory MPC model in $O(\pll(n))$ rounds.
First, because the running time in Theorem~\ref{thm:betarsMPC} part (i) has an inverse exponential dependency on $\beta$, we get the following corollary.

\begin{corollary}
For  $\beta \in \Omega(\log\log\log \Delta)$, a $\beta$-ruling set of a graph $G$ can be computed in $O(\beta \log\log n)$ rounds whp in the low-memory MPC model with $\tilde{O}(m+n^{1+\eps})$ total memory.
\end{corollary}

Second, we can also show that for graphs with bounded $\Delta$, we can compute a $\beta$-ruling set in $O(\beta \log\log n)$ rounds, however this bound increases quickly with $\beta$, giving us the following corollary.

\begin{corollary}
If we have that $\Delta = O\left(2^{\log^{1-\frac{1}{2^{\beta}}} n}\right)$, then a $\beta$-ruling set can be computed in $O(\beta \log\log n)$ rounds whp in the low-memory MPC model with $\tilde{O}(m+n^{1+\eps})$ total memory.
\end{corollary}  
\ifhideproofs{}\begin{proof}
We set $f_i := 2^{\log^{\delta_i} \Delta}$, \(\delta_{i} = 1-\frac{1}{2^{\beta-i}}\) for all $i$, $0 \le i \le \beta -1$ in Theorem~\ref{thm:betarsMPC} part (i). Note that in this case we set \(f_0 = \Delta = O(2^{\log^{1-\frac{1}{2^{\beta}}} n})\), therefore each of the \(\beta\) terms it the running time containing \(f_i\)'s becomes a constant. This leads to a running time of \(O(\beta \log\log n)\).
\end{proof}\fi{} 

\section{Conclusions and Future Work}
The results we developed in this paper show that 2-ruling sets can be computed much faster than an MIS in the low-memory MPC model. In the absence of explicit lower bounds for this problem in the low-memory MPC model, it is an open question if we can improve on the round complexity of $O(\log^{1/6} \Delta)$ for computing a 2-ruling set in the low-memory MPC model. Another aspect to note is the vast difference in the runtime of the $\beta$-ruling set algorithms in the two  settings we consider with respect to the total memory.  
It is not clear if this divergence is natural due to the restriction in the model, or indicates a scope for improving Theorem \ref{theorem:mTotalMemory}.

 \bibliographystyle{plainurl}
\bibliography{refs}

\begin{thebibliography}{10}

\bibitem{AhnGSPAA15}
Kook~Jin Ahn and Sudipto Guha.
\newblock Access to data and number of iterations: Dual primal algorithms for
  maximum matching under resource constraints.
\newblock In {\em Proceedings of the 27th ACM Symposium on Parallelism in
  Algorithms and Architectures}, SPAA ’15, page 202–211, New York, NY, USA,
  2015. Association for Computing Machinery.
\newblock URL: \url{https://doi.org/10.1145/2755573.2755586}, \href
  {http://dx.doi.org/10.1145/2755573.2755586}
  {\path{doi:10.1145/2755573.2755586}}.

\bibitem{AlonetalJAlg1986}
Noga Alon, L\'{a}szl\'{o} Babai, and Alon Itai.
\newblock A fast and simple randomized parallel algorithm for the maximal
  independent set problem.
\newblock {\em J. Algorithms}, 7(4):567–583, December 1986.
\newblock URL: \url{https://doi.org/10.1016/0196-6774(86)90019-2}, \href
  {http://dx.doi.org/10.1016/0196-6774(86)90019-2}
  {\path{doi:10.1016/0196-6774(86)90019-2}}.

\bibitem{AndoniNOY14}
Alexandr Andoni, Aleksandar Nikolov, Krzysztof Onak, and Grigory Yaroslavtsev.
\newblock Parallel algorithms for geometric graph problems.
\newblock In David~B. Shmoys, editor, {\em Symposium on Theory of Computing,
  {STOC} 2014, New York, NY, USA, May 31 - June 03, 2014}, pages 574--583.
  {ACM}, 2014.
\newblock URL: \url{https://doi.org/10.1145/2591796.2591805}, \href
  {http://dx.doi.org/10.1145/2591796.2591805}
  {\path{doi:10.1145/2591796.2591805}}.

\bibitem{Andoni2018}
Alexandr Andoni, Zhao Song, Clifford Stein, Zhengyu Wang, and Peilin Zhong.
\newblock {Parallel graph connectivity in log diameter rounds}.
\newblock {\em Proceedings - Annual IEEE Symposium on Foundations of Computer
  Science, FOCS}, 2018-October:674--685, 2018.
\newblock \href {http://arxiv.org/abs/1805.03055} {\path{arXiv:1805.03055}},
  \href {http://dx.doi.org/10.1109/FOCS.2018.00070}
  {\path{doi:10.1109/FOCS.2018.00070}}.

\bibitem{AndoniSZ20}
Alexandr Andoni, Clifford Stein, and Peilin Zhong.
\newblock Parallel approximate undirected shortest paths via low hop emulators.
\newblock In Konstantin Makarychev, Yury Makarychev, Madhur Tulsiani, Gautam
  Kamath, and Julia Chuzhoy, editors, {\em Proccedings of the 52nd Annual {ACM}
  {SIGACT} Symposium on Theory of Computing, {STOC} 2020, Chicago, IL, USA,
  June 22-26, 2020}, pages 322--335. {ACM}, 2020.
\newblock URL: \url{https://doi.org/10.1145/3357713.3384321}, \href
  {http://dx.doi.org/10.1145/3357713.3384321}
  {\path{doi:10.1145/3357713.3384321}}.

\bibitem{Assadi2017a}
Sepehr Assadi.
\newblock Simple round compression for parallel vertex cover.
\newblock {\em CoRR}, abs/1709.04599, 2017.
\newblock URL: \url{http://arxiv.org/abs/1709.04599}, \href
  {http://arxiv.org/abs/1709.04599} {\path{arXiv:1709.04599}}.

\bibitem{Assadi2018b}
Sepehr Assadi, Yu~Chen, and Sanjeev Khanna.
\newblock {Sublinear algorithms for $(\Delta + 1)$ vertex coloring}.
\newblock {\em Proceedings of the Annual ACM-SIAM Symposium on Discrete
  Algorithms}, pages 767--786, 2019.
\newblock URL: \url{http://arxiv.org/abs/1807.08886}, \href
  {http://arxiv.org/abs/1807.08886} {\path{arXiv:1807.08886}}, \href
  {http://dx.doi.org/10.1137/1.9781611975482.48}
  {\path{doi:10.1137/1.9781611975482.48}}.

\bibitem{AssadiKZPODC19}
Sepehr Assadi, Nikolai Karpov, and Qin Zhang.
\newblock Distributed and streaming linear programming in low dimensions.
\newblock In Dan Suciu, Sebastian Skritek, and Christoph Koch, editors, {\em
  Proceedings of the 38th {ACM} {SIGMOD-SIGACT-SIGAI} Symposium on Principles
  of Database Systems, {PODS} 2019, Amsterdam, The Netherlands, June 30 - July
  5, 2019}, pages 236--253. {ACM}, 2019.
\newblock URL: \url{https://doi.org/10.1145/3294052.3319697}, \href
  {http://dx.doi.org/10.1145/3294052.3319697}
  {\path{doi:10.1145/3294052.3319697}}.

\bibitem{Assadi2018a}
Sepehr Assadi, Xiaorui Sun, and Omri Weinstein.
\newblock {Massively parallel algorithms for finding well-connected components
  in sparse graphs}.
\newblock {\em Proceedings of the Annual ACM Symposium on Principles of
  Distributed Computing}, pages 461--470, 2019.
\newblock \href {http://arxiv.org/abs/1805.02974} {\path{arXiv:1805.02974}},
  \href {http://dx.doi.org/10.1145/3293611.3331596}
  {\path{doi:10.1145/3293611.3331596}}.

\bibitem{AssadiSWPODC19}
Sepehr Assadi, Xiaorui Sun, and Omri Weinstein.
\newblock Massively parallel algorithms for finding well-connected components
  in sparse graphs.
\newblock In Peter Robinson and Faith Ellen, editors, {\em Proceedings of the
  2019 {ACM} Symposium on Principles of Distributed Computing, {PODC} 2019,
  Toronto, ON, Canada, July 29 - August 2, 2019}, pages 461--470. {ACM}, 2019.
\newblock URL: \url{https://doi.org/10.1145/3293611.3331596}, \href
  {http://dx.doi.org/10.1145/3293611.3331596}
  {\path{doi:10.1145/3293611.3331596}}.

\bibitem{BEPS16}
Leonid Barenboim, Michael Elkin, Seth Pettie, and Johannes Schneider.
\newblock The locality of distributed symmetry breaking.
\newblock {\em J. {ACM}}, 63(3):20:1--20:45, 2016.
\newblock URL: \url{https://doi.org/10.1145/2903137}, \href
  {http://dx.doi.org/10.1145/2903137} {\path{doi:10.1145/2903137}}.

\bibitem{BeameKS14}
Paul Beame, Paraschos Koutris, and Dan Suciu.
\newblock Skew in parallel query processing.
\newblock In Richard Hull and Martin Grohe, editors, {\em Proceedings of the
  33rd {ACM} {SIGMOD-SIGACT-SIGART} Symposium on Principles of Database
  Systems, PODS'14, Snowbird, UT, USA, June 22-27, 2014}, pages 212--223.
  {ACM}, 2014.
\newblock URL: \url{https://doi.org/10.1145/2594538.2594558}, \href
  {http://dx.doi.org/10.1145/2594538.2594558}
  {\path{doi:10.1145/2594538.2594558}}.

\bibitem{BeameKS17}
Paul Beame, Paraschos Koutris, and Dan Suciu.
\newblock Communication steps for parallel query processing.
\newblock {\em J. {ACM}}, 64(6):40:1--40:58, 2017.
\newblock URL: \url{https://doi.org/10.1145/3125644}, \href
  {http://dx.doi.org/10.1145/3125644} {\path{doi:10.1145/3125644}}.

\bibitem{BehnezhadetalPODC2019}
Soheil Behnezhad, Sebastian Brandt, Mahsa Derakhshan, Manuela Fischer,
  MohammadTaghi Hajiaghayi, Richard~M. Karp, and Jara Uitto.
\newblock Massively parallel computation of matching and mis in sparse graphs.
\newblock In {\em Proceedings of the 2019 ACM Symposium on Principles of
  Distributed Computing}, PODC ’19, page 481–490, New York, NY, USA, 2019.
  Association for Computing Machinery.
\newblock URL: \url{https://doi.org/10.1145/3293611.3331609}, \href
  {http://dx.doi.org/10.1145/3293611.3331609}
  {\path{doi:10.1145/3293611.3331609}}.

\bibitem{BehnezhadDH18}
Soheil Behnezhad, Mahsa Derakhshan, and MohammadTaghi Hajiaghayi.
\newblock Brief announcement: Semi-mapreduce meets congested clique.
\newblock {\em CoRR}, abs/1802.10297, 2018.
\newblock URL: \url{http://arxiv.org/abs/1802.10297}, \href
  {http://arxiv.org/abs/1802.10297} {\path{arXiv:1802.10297}}.

\bibitem{Behnezhad2019}
Soheil Behnezhad, Mohammadtaghi Hajiaghayi, and David~G. Harris.
\newblock {Exponentially Faster Massively Parallel Maximal Matching}.
\newblock Technical report, 2019.
\newblock \href {http://arxiv.org/abs/1901.03744} {\path{arXiv:1901.03744}},
  \href {http://dx.doi.org/10.1109/FOCS.2019.00096}
  {\path{doi:10.1109/FOCS.2019.00096}}.

\bibitem{BishtKP13}
Tushar Bisht, Kishore Kothapalli, and Sriram~V. Pemmaraju.
\newblock Brief announcement: Super-fast t-ruling sets.
\newblock In Magn{\'{u}}s~M. Halld{\'{o}}rsson and Shlomi Dolev, editors, {\em
  {ACM} Symposium on Principles of Distributed Computing, {PODC} '14, Paris,
  France, July 15-18, 2014}, pages 379--381. {ACM}, 2014.
\newblock URL: \url{https://doi.org/10.1145/2611462.2611512}, \href
  {http://dx.doi.org/10.1145/2611462.2611512}
  {\path{doi:10.1145/2611462.2611512}}.

\bibitem{BoroujeniEGHS18}
Mahdi Boroujeni, Soheil Ehsani, Mohammad Ghodsi, Mohammad~Taghi Hajiaghayi, and
  Saeed Seddighin.
\newblock Approximating edit distance in truly subquadratic time: Quantum and
  mapreduce.
\newblock In Artur Czumaj, editor, {\em Proceedings of the Twenty-Ninth Annual
  {ACM-SIAM} Symposium on Discrete Algorithms, {SODA} 2018, New Orleans, LA,
  USA, January 7-10, 2018}, pages 1170--1189. {SIAM}, 2018.
\newblock URL: \url{https://doi.org/10.1137/1.9781611975031.76}, \href
  {http://dx.doi.org/10.1137/1.9781611975031.76}
  {\path{doi:10.1137/1.9781611975031.76}}.

\bibitem{ChangetalPODC2019}
Yi-Jun Chang, Manuela Fischer, Mohsen Ghaffari, Jara Uitto, and Yufan Zheng.
\newblock The complexity of \((\delta+1)\) coloring in congested clique,
  massively parallel computation, and centralized local computation.
\newblock In {\em Proceedings of the 2019 ACM Symposium on Principles of
  Distributed Computing}, PODC ’19, page 471–480, New York, NY, USA, 2019.
  Association for Computing Machinery.
\newblock URL: \url{https://doi.org/10.1145/3293611.3331607}, \href
  {http://dx.doi.org/10.1145/3293611.3331607}
  {\path{doi:10.1145/3293611.3331607}}.

\bibitem{Charikar2020}
Moses Charikar, Weiyun Ma, and Li-Yang Tan.
\newblock Unconditional lower bounds for adaptive massively parallel
  computation.
\newblock In {\em Proceedings of the 32nd ACM Symposium on Parallelism in
  Algorithms and Architectures}, SPAA ’20, page 141–151, New York, NY, USA,
  2020. Association for Computing Machinery.
\newblock URL: \url{https://doi.org/10.1145/3350755.3400230}, \href
  {http://dx.doi.org/10.1145/3350755.3400230}
  {\path{doi:10.1145/3350755.3400230}}.

\bibitem{ChingetalVLDB2015}
Avery Ching, Sergey Edunov, Maja Kabiljo, Dionysios Logothetis, and Sambavi
  Muthukrishnan.
\newblock One trillion edges: Graph processing at facebook-scale.
\newblock {\em Proc. VLDB Endow.}, 8(12):1804–1815, August 2015.
\newblock URL: \url{https://doi.org/10.14778/2824032.2824077}, \href
  {http://dx.doi.org/10.14778/2824032.2824077}
  {\path{doi:10.14778/2824032.2824077}}.

\bibitem{Czumaj2019}
Artur Czumaj, Peter Davies, and Merav Parter.
\newblock {Graph Sparsification for Derandomizing Massively Parallel
  Computation with Low Space}.
\newblock Technical report, 2019.
\newblock URL: \url{http://arxiv.org/abs/1912.05390}, \href
  {http://arxiv.org/abs/1912.05390} {\path{arXiv:1912.05390}}.

\bibitem{Czumaj2017}
Artur Czumaj, Slobodan Mitrovic, Jakub {\L}{\c{a}}cki, Krzysztof Onak,
  Aleksander M{\c{a}}dry, and Piotr Sankowski.
\newblock {Round compression for parallel matching algorithms}.
\newblock {\em Proceedings of the Annual ACM Symposium on Theory of Computing},
  (1):471--484, 2018.
\newblock URL: \url{http://arxiv.org/abs/1707.03478}, \href
  {http://arxiv.org/abs/1707.03478} {\path{arXiv:1707.03478}}, \href
  {http://dx.doi.org/10.1145/3188745.3188764}
  {\path{doi:10.1145/3188745.3188764}}.

\bibitem{DeanG08}
Jeffrey Dean and Sanjay Ghemawat.
\newblock Mapreduce: Simplified data processing on large clusters.
\newblock {\em Commun. ACM}, 51(1):107–113, January 2008.
\newblock URL: \url{https://doi.org/10.1145/1327452.1327492}, \href
  {http://dx.doi.org/10.1145/1327452.1327492}
  {\path{doi:10.1145/1327452.1327492}}.

\bibitem{GhaffariSODA2016}
Mohsen Ghaffari.
\newblock An improved distributed algorithm for maximal independent set.
\newblock In {\em Proceedings of the Twenty-Seventh Annual ACM-SIAM Symposium
  on Discrete Algorithms}, SODA ’16, page 270–277, USA, 2016. Society for
  Industrial and Applied Mathematics.

\bibitem{Ghaffari2017}
Mohsen Ghaffari.
\newblock {Distributed MIS via all-to-all communication}.
\newblock {\em Proceedings of the Annual ACM Symposium on Principles of
  Distributed Computing}, Part F129314:141--150, 2017.
\newblock \href {http://dx.doi.org/10.1145/3087801.3087830}
  {\path{doi:10.1145/3087801.3087830}}.

\bibitem{Ghaffari2018a}
Mohsen Ghaffari, Themis Gouleakis, Christian Konrad, Slobodan Mitrovi{\'{c}},
  and Ronitt Rubinfeld.
\newblock {Improved massively parallel computation algorithms for MIS,
  matching, and vertex cover}.
\newblock {\em Proceedings of the Annual ACM Symposium on Principles of
  Distributed Computing}, pages 129--138, 2018.
\newblock \href {http://arxiv.org/abs/1802.08237} {\path{arXiv:1802.08237}},
  \href {http://dx.doi.org/10.1145/3212734.3212743}
  {\path{doi:10.1145/3212734.3212743}}.

\bibitem{Ghaffari2020}
Mohsen Ghaffari, Christoph Grunau, and Ce~Jin.
\newblock {Improved MPC Algorithms for MIS, Matching, and Coloring on Trees and
  Beyond}.
\newblock feb 2020.
\newblock URL: \url{http://arxiv.org/abs/2002.09610}, \href
  {http://arxiv.org/abs/2002.09610} {\path{arXiv:2002.09610}}.

\bibitem{GhaffariJNSPAA20}
Mohsen Ghaffari, Ce~Jin, and Daan Nilis.
\newblock A massively parallel algorithm for minimum weight vertex cover.
\newblock In Christian Scheideler and Michael Spear, editors, {\em {SPAA} '20:
  32nd {ACM} Symposium on Parallelism in Algorithms and Architectures, Virtual
  Event, USA, July 15-17, 2020}, pages 259--268. {ACM}, 2020.
\newblock URL: \url{https://doi.org/10.1145/3350755.3400260}, \href
  {http://dx.doi.org/10.1145/3350755.3400260}
  {\path{doi:10.1145/3350755.3400260}}.

\bibitem{Ghaffari2019}
Mohsen Ghaffari, Fabian Kuhn, and Jara Uitto.
\newblock {Conditional hardness results for massively parallel computation from
  distributed lower bounds}.
\newblock Technical report, 2019.
\newblock \href {http://dx.doi.org/10.1109/FOCS.2019.00097}
  {\path{doi:10.1109/FOCS.2019.00097}}.

\bibitem{Ghaffari2018}
Mohsen Ghaffari and Jara Uitto.
\newblock {Sparsifying distributed algorithms with ramifications in massively
  parallel computation and centralized local computation}.
\newblock {\em Proceedings of the Annual ACM-SIAM Symposium on Discrete
  Algorithms}, pages 1636--1653, jul 2019.
\newblock URL: \url{http://arxiv.org/abs/1807.06251}, \href
  {http://arxiv.org/abs/1807.06251} {\path{arXiv:1807.06251}}, \href
  {http://dx.doi.org/10.1137/1.9781611975482.99}
  {\path{doi:10.1137/1.9781611975482.99}}.

\bibitem{Goodrich2011}
Michael~T. Goodrich, Nodari Sitchinava, and Qin Zhang.
\newblock {Sorting, searching, and simulation in the MapReduce framework}.
\newblock {\em Lecture Notes in Computer Science (including subseries Lecture
  Notes in Artificial Intelligence and Lecture Notes in Bioinformatics)}, 7074
  LNCS:374--383, 2011.
\newblock \href {http://arxiv.org/abs/1101.1902} {\path{arXiv:1101.1902}},
  \href {http://dx.doi.org/10.1007/978-3-642-25591-5_39}
  {\path{doi:10.1007/978-3-642-25591-5_39}}.

\bibitem{HarveyetalSPAA2018}
Nicholas J.~A. Harvey, Christopher Liaw, and Paul Liu.
\newblock Greedy and local ratio algorithms in the mapreduce model.
\newblock In {\em Proceedings of the 30th ACM Symposium on Parallelism in
  Algorithms and Architectures}, SPAA ’18, page 43–52, New York, NY, USA,
  2018. Association for Computing Machinery.
\newblock URL: \url{https://doi.org/10.1145/3210377.3210386}, \href
  {http://dx.doi.org/10.1145/3210377.3210386}
  {\path{doi:10.1145/3210377.3210386}}.

\bibitem{HegemanPTCS15}
James~W. Hegeman and Sriram~V. Pemmaraju.
\newblock Lessons from the congested clique applied to mapreduce.
\newblock {\em Theor. Comput. Sci.}, 608:268--281, 2015.
\newblock URL: \url{https://doi.org/10.1016/j.tcs.2015.09.029}, \href
  {http://dx.doi.org/10.1016/j.tcs.2015.09.029}
  {\path{doi:10.1016/j.tcs.2015.09.029}}.

\bibitem{HegemanPSarxiv2014}
James~W. Hegeman, Sriram~V. Pemmaraju, and Vivek Sardeshmukh.
\newblock Near-constant-time distributed algorithms on a congested clique.
\newblock {\em CoRR}, abs/1408.2071, 2014.
\newblock URL: \url{http://arxiv.org/abs/1408.2071}, \href
  {http://arxiv.org/abs/1408.2071} {\path{arXiv:1408.2071}}.

\bibitem{ImBSSTOC17}
Sungjin Im, Benjamin Moseley, and Xiaorui Sun.
\newblock Efficient massively parallel methods for dynamic programming.
\newblock In {\em Proceedings of the 49th Annual ACM SIGACT Symposium on Theory
  of Computing}, STOC 2017, page 798–811, New York, NY, USA, 2017.
  Association for Computing Machinery.
\newblock URL: \url{https://doi.org/10.1145/3055399.3055460}, \href
  {http://dx.doi.org/10.1145/3055399.3055460}
  {\path{doi:10.1145/3055399.3055460}}.

\bibitem{Inamdar2018}
Tanmay Inamdar, Shreyas Pai, and Sriram~V. Pemmaraju.
\newblock {Large-scale distributed algorithms for facility location with
  outliers}.
\newblock {\em 22nd International Conference on Principles of Distributed
  Systems (OPODIS 2018)}, 125, nov 2019.
\newblock URL: \url{http://arxiv.org/abs/1811.06494}, \href
  {http://arxiv.org/abs/1811.06494} {\path{arXiv:1811.06494}}, \href
  {http://dx.doi.org/10.4230/LIPIcs.OPODIS.2018.5}
  {\path{doi:10.4230/LIPIcs.OPODIS.2018.5}}.

\bibitem{Karloff2010}
Howard Karloff, Siddharth Suri, and Sergei Vassilvitskii.
\newblock {A Model of Computation for MapReduce}.
\newblock In {\em Proceedings of the Twenty-First Annual ACM-SIAM Symposium on
  Discrete Algorithms}, pages 938--948, Philadelphia, PA, jan 2010. Society for
  Industrial and Applied Mathematics.
\newblock URL: \url{https://epubs.siam.org/doi/10.1137/1.9781611973075.76},
  \href {http://dx.doi.org/10.1137/1.9781611973075.76}
  {\path{doi:10.1137/1.9781611973075.76}}.

\bibitem{KlauckNPRSODA15}
Hartmut Klauck, Danupon Nanongkai, Gopal Pandurangan, and Peter Robinson.
\newblock Distributed computation of large-scale graph problems.
\newblock In Piotr Indyk, editor, {\em Proceedings of the Twenty-Sixth Annual
  {ACM-SIAM} Symposium on Discrete Algorithms, {SODA} 2015, San Diego, CA, USA,
  January 4-6, 2015}, pages 391--410. {SIAM}, 2015.
\newblock URL: \url{https://doi.org/10.1137/1.9781611973730.28}, \href
  {http://dx.doi.org/10.1137/1.9781611973730.28}
  {\path{doi:10.1137/1.9781611973730.28}}.

\bibitem{Konrad}
Christian Konrad, Sriram~V. Pemmaraju, Talal Riaz, and Peter Robinson.
\newblock The complexity of symmetry breaking in massive graphs.
\newblock In Jukka Suomela, editor, {\em 33rd International Symposium on
  Distributed Computing, {DISC} 2019, October 14-18, 2019, Budapest, Hungary},
  volume 146, pages 26:1--26:18, 2019.

\bibitem{KothapalliP12}
Kishore Kothapalli and Sriram~V. Pemmaraju.
\newblock Super-fast 3-ruling sets.
\newblock In Deepak D'Souza, Telikepalli Kavitha, and Jaikumar Radhakrishnan,
  editors, {\em {IARCS} Annual Conference on Foundations of Software Technology
  and Theoretical Computer Science, {FSTTCS} 2012, December 15-17, 2012,
  Hyderabad, India}, volume~18 of {\em LIPIcs}, pages 136--147. Schloss
  Dagstuhl - Leibniz-Zentrum f{\"{u}}r Informatik, 2012.
\newblock URL: \url{https://doi.org/10.4230/LIPIcs.FSTTCS.2012.136}, \href
  {http://dx.doi.org/10.4230/LIPIcs.FSTTCS.2012.136}
  {\path{doi:10.4230/LIPIcs.FSTTCS.2012.136}}.

\bibitem{KoutrisBS16}
Paraschos Koutris, Paul Beame, and Dan Suciu.
\newblock Worst-case optimal algorithms for parallel query processing.
\newblock In Wim Martens and Thomas Zeume, editors, {\em 19th International
  Conference on Database Theory, {ICDT} 2016, Bordeaux, France, March 15-18,
  2016}, volume~48 of {\em LIPIcs}, pages 8:1--8:18. Schloss Dagstuhl -
  Leibniz-Zentrum f{\"{u}}r Informatik, 2016.
\newblock URL: \url{https://doi.org/10.4230/LIPIcs.ICDT.2016.8}, \href
  {http://dx.doi.org/10.4230/LIPIcs.ICDT.2016.8}
  {\path{doi:10.4230/LIPIcs.ICDT.2016.8}}.

\bibitem{KMWPODC2004}
Fabian Kuhn, Thomas Moscibroda, and Roger Wattenhofer.
\newblock What cannot be computed locally!
\newblock In {\em Proceedings of the Twenty-Third Annual ACM Symposium on
  Principles of Distributed Computing}, PODC ’04, page 300–309, New York,
  NY, USA, 2004. Association for Computing Machinery.
\newblock URL: \url{https://doi.org/10.1145/1011767.1011811}, \href
  {http://dx.doi.org/10.1145/1011767.1011811}
  {\path{doi:10.1145/1011767.1011811}}.

\bibitem{KMWJACM2016}
Fabian Kuhn, Thomas Moscibroda, and Roger Wattenhofer.
\newblock Local computation: Lower and upper bounds.
\newblock {\em J. ACM}, 63(2), March 2016.
\newblock URL: \url{https://doi.org/10.1145/2742012}, \href
  {http://dx.doi.org/10.1145/2742012} {\path{doi:10.1145/2742012}}.

\bibitem{LattanziMSVSPAA11}
Silvio Lattanzi, Benjamin Moseley, Siddharth Suri, and Sergei Vassilvitskii.
\newblock Filtering: A method for solving graph problems in mapreduce.
\newblock In {\em Proceedings of the Twenty-Third Annual ACM Symposium on
  Parallelism in Algorithms and Architectures}, SPAA ’11, page 85–94, New
  York, NY, USA, 2011. Association for Computing Machinery.
\newblock URL: \url{https://doi.org/10.1145/1989493.1989505}, \href
  {http://dx.doi.org/10.1145/1989493.1989505}
  {\path{doi:10.1145/1989493.1989505}}.

\bibitem{Linial92}
Nathan Linial.
\newblock Locality in distributed graph algorithms.
\newblock {\em SIAM J. Comput.}, 21(1):193–201, February 1992.
\newblock URL: \url{https://doi.org/10.1137/0221015}, \href
  {http://dx.doi.org/10.1137/0221015} {\path{doi:10.1137/0221015}}.

\bibitem{Peleg03}
Zvi Lotker, Elan Pavlov, Boaz Patt{-}Shamir, and David Peleg.
\newblock {MST} construction in o(log log n) communication rounds.
\newblock In {\em {SPAA}}, pages 94--100, 2003.
\newblock URL: \url{https://doi.org/10.1145/777412.777428}, \href
  {http://dx.doi.org/10.1145/777412.777428} {\path{doi:10.1145/777412.777428}}.

\bibitem{LubySICOMP1986}
Michael Luby.
\newblock A simple parallel algorithm for the maximal independent set problem.
\newblock {\em SIAM Journal on Computing}, 15(4):1036--1053, 1986.
\newblock URL: \url{https://doi.org/10.1137/0215074}, \href
  {http://arxiv.org/abs/https://doi.org/10.1137/0215074}
  {\path{arXiv:https://doi.org/10.1137/0215074}}, \href
  {http://dx.doi.org/10.1137/0215074} {\path{doi:10.1137/0215074}}.

\bibitem{pregel}
Grzegorz Malewicz, Matthew~H. Austern, Aart~J.C Bik, James~C. Dehnert, Ilan
  Horn, Naty Leiser, and Grzegorz Czajkowski.
\newblock Pregel: A system for large-scale graph processing.
\newblock In {\em Proceedings of the 2010 ACM SIGMOD International Conference
  on Management of Data}, SIGMOD ’10, page 135–146, New York, NY, USA,
  2010. Association for Computing Machinery.
\newblock URL: \url{https://doi.org/10.1145/1807167.1807184}, \href
  {http://dx.doi.org/10.1145/1807167.1807184}
  {\path{doi:10.1145/1807167.1807184}}.

\bibitem{ParterYogevDISC2018}
Merav Parter and Eylon Yogev.
\newblock Congested clique algorithms for graph spanners.
\newblock In Ulrich Schmid and Josef Widder, editors, {\em 32nd International
  Symposium on Distributed Computing, {DISC} 2018, New Orleans, LA, USA,
  October 15-19, 2018}, volume 121 of {\em LIPIcs}, pages 40:1--40:18. Schloss
  Dagstuhl - Leibniz-Zentrum f{\"{u}}r Informatik, 2018.
\newblock URL: \url{https://doi.org/10.4230/LIPIcs.DISC.2018.40}, \href
  {http://dx.doi.org/10.4230/LIPIcs.DISC.2018.40}
  {\path{doi:10.4230/LIPIcs.DISC.2018.40}}.

\bibitem{peleg00}
D.~Peleg.
\newblock {\em Distributed Computing: A Locality-Sensitive Approach}.
\newblock SIAM, 2000.

\bibitem{RoughgardenVWJACM2018}
Tim Roughgarden, Sergei Vassilvitskii, and Joshua~R. Wang.
\newblock Shuffles and circuits (on lower bounds for modern parallel
  computation).
\newblock {\em J. ACM}, 65(6), November 2018.
\newblock URL: \url{https://doi.org/10.1145/3232536}, \href
  {http://dx.doi.org/10.1145/3232536} {\path{doi:10.1145/3232536}}.

\bibitem{RozhonGSTOC20}
V{\'{a}}clav Rozhon and Mohsen Ghaffari.
\newblock Polylogarithmic-time deterministic network decomposition and
  distributed derandomization.
\newblock In {\em Proccedings of the 52nd Annual {ACM} {SIGACT} Symposium on
  Theory of Computing, {STOC} 2020, Chicago, IL, USA, June 22-26, 2020}, pages
  350--363, 2020.
\newblock URL: \url{https://doi.org/10.1145/3357713.3384298}, \href
  {http://dx.doi.org/10.1145/3357713.3384298}
  {\path{doi:10.1145/3357713.3384298}}.

\bibitem{ShankhaBiswas2020}
Amartya {Shankha Biswas}, Talya Eden, Quanquan~C Liu, Slobodan Mitrovi, and
  Ronitt Rubinfeld.
\newblock {Parallel Algorithms for Small Subgraph Counting}.
\newblock Technical report, 2020.
\newblock \href {http://arxiv.org/abs/2002.08299v1}
  {\path{arXiv:2002.08299v1}}.

\bibitem{YaroslavtsevVadapalliICML2018}
Grigory Yaroslavtsev and Adithya Vadapalli.
\newblock Massively parallel algorithms and hardness for single-linkage
  clustering under \({\ell}_{\mbox{p}}\) distances.
\newblock In Jennifer~G. Dy and Andreas Krause, editors, {\em Proceedings of
  the 35th International Conference on Machine Learning, {ICML} 2018,
  Stockholmsm{\"{a}}ssan, Stockholm, Sweden, July 10-15, 2018}, volume~80 of
  {\em Proceedings of Machine Learning Research}, pages 5596--5605. {PMLR},
  2018.
\newblock URL: \url{http://proceedings.mlr.press/v80/yaroslavtsev18a.html}.

\bibitem{spark}
Matei Zaharia, Reynold~S. Xin, Patrick Wendell, Tathagata Das, Michael
  Armbrust, Ankur Dave, Xiangrui Meng, Josh Rosen, Shivaram Venkataraman,
  Michael~J. Franklin, Ali Ghodsi, Joseph Gonzalez, Scott Shenker, and Ion
  Stoica.
\newblock Apache spark: A unified engine for big data processing.
\newblock {\em Commun. ACM}, 59(11):56–65, October 2016.
\newblock URL: \url{https://doi.org/10.1145/2934664}, \href
  {http://dx.doi.org/10.1145/2934664} {\path{doi:10.1145/2934664}}.

\end{thebibliography}

\end{document}